\documentclass{fundam}

\usepackage[utf8]{inputenc}
\usepackage{graphicx}
\usepackage{amssymb}
\usepackage{amsmath}
\usepackage{url}
\usepackage{lineno}
\usepackage[linesnumbered,lined,boxed,commentsnumbered]{algorithm2e}

\def\Z{\mathbb{Z}}

\def\C{\mathcal{C}}
\def\F{\mathcal{R}}
\def\F{\mathcal{F}}
\def\A{\mathcal{A}}
\def\P{\mathcal{P}}
\def\W{\mathcal{W}}

\def\H{\mathcal{H}}
\def\V{\mathcal{V}}
\def\convv{\mathrm{conv _{R^2}}}
\def\conv{\mathrm{conv _{R^d}}}
\def\Right{\mathrm{Right}}
\def\Left{\mathrm{Left}}
\def\South{\mathrm{South}}
\def\East{\mathrm{East}}
\def\North{\mathrm{North}}
\def\West{\mathrm{West}}

\def\In{\mathrm{Kernel}}
\def\Out{\mathrm{Out}}
\def\conv{\mathrm{conv} _{Z ^2}}
\def\Undetermined{\mathrm{Undetermined}}

\def\SE{\mathrm{SE}}
\def\NE{\mathrm{NE}}
\def\NW{\mathrm{NW}}
\def\SW{\mathrm{SW}}
\def\SEV{\mathrm{SEV}}
\def\NEV{\mathrm{NEV}}
\def\NWV{\mathrm{NWV}}
\def\SWV{\mathrm{SWV}}
\def\End{\mathrm{End}}
\def\Start{\mathrm{Start}}
\def\Left{\mathrm{Left}}
\def\Right{\mathrm{Right}}
\def\Hstrip{\mathrm{Hstrip}}
\def\Vstrip{\mathrm{Vstrip}}

\makeatletter
\newcommand\bigDiamond{\mathop{\mathpalette\bigDi@mond\relax}}
\newcommand\bigDi@mond[2]{%
  \vcenter{\hbox{\m@th
    \scalebox{\ifx#1\displaystyle 2\else1.2\fi}{$#1\Diamond$}%
  }}%
}

\newtheorem{pb}{Problem}

\begin{document}

\setcounter{page}{113}
\publyear{22}
\papernumber{2155}
\volume{189}
\issue{2}

  \finalVersionForARXIV

\title{Reconstruction of Convex Sets from One or Two X-rays} 

\author{Yan Gerard\thanks{Address  for correspondence:  Universit\'e Clermont Auvergne,
LIMOS, France}\thanks{This work is supported by the French ANR PRC grant ADDS (ANR-19-CE48-0005)}
\\
Universit\'e Clermont Auvergne, LIMOS, France\\
yan.gerard@uca.fr}

\maketitle

\runninghead{Y. Gerard}{Reconstruction of Convex Sets}

\begin{abstract}
	We consider a class of problems of Discrete Tomography which has been deeply investigated in the past: the reconstruction of convex lattice sets from their horizontal and/or vertical X-rays, i.e. from the number of points in a sequence of consecutive horizontal and vertical lines.
The reconstruction of the HV-convex polyominoes works usually in two steps, first the filling step consisting in filling operations, second the convex aggregation of the switching components.
We prove three results about the convex aggregation step:
(1) The convex aggregation step used for the reconstruction of HV-convex polyominoes does not always provide a solution. The example yielding to this result is called \textit{the bad guy} and disproves a conjecture of the domain.
(2) The reconstruction of a digital convex lattice set from only one X-ray can be performed in polynomial time. We prove it by encoding the convex aggregation problem in a Directed Acyclic Graph. (3) With the same strategy, we prove that the reconstruction of fat digital convex sets from their horizontal and vertical X-rays can be solved in polynomial time. Fatness is a property of the digital convex sets regarding the relative position of the left, right, top and bottom points of the set.
The complexity of the reconstruction of the digital convex sets which are not fat remains an open question.
\end{abstract}

\begin{keywords}
Discrete Tomography \and Polyomino \and Digital Convex sets \and Filling operations \and Switching Component
\end{keywords}

\section{Introduction} \label{S:1}

\subsection{About discrete tomography}

In the mid 1990s, researchers in Material Science and especially in three dimensional Electron Microscopy previewed the development of an upcoming technology  able to count the number of atoms of a material crossed by a beam of straight lines \cite{EM}. According to the same strategy than Computerized Tomography, they intended to use this process in order to reconstruct the 3D structure of different  materials (proteins, crystals...) with a very high level of precision. In this new framework, the classical algorithms of Computerized Tomography have been unable to provide satisfying results. The CT algorithms have been designed for the investigation of materials at a scale where it can be assumed to be continuous. These algorithms are very poorly adapted for a level where the set of atoms is closer to a discrete set of points. The discrete nature of the objects to be reconstructed  is the first difficulty which makes CT algorithms ineffective at the atomic scale. A second difficulty comes from the very low number of X-rays -from $2$ to $10$- which can be used in Material Science since the X-rays damage the atomic structure (as comparison, CT-scans provide usually hundreds of X-rays). The third difference with Computerized Tomography is that for the reconstruction of the atomic structure of crystals (see \cite{Bat1,Bat2} for crystalline structures of nano-particules computed with Discrete Tomography in the 2010s years), the atoms are centered on a lattice so that the problem becomes the reconstruction of a lattice set, namely in dimension $2$ a binary matrix.

The development of the technology for counting the number of atoms on straight lines took finally more time than expected  but the impulse was given to explore this new range of questions dealing with the reconstruction of discrete sets of points. The sequence of the numbers of points of the intersections of a discrete set with consecutive parallel lines has been called by keeping the physical term of \textit{X-ray} while the reconstruction of a discrete set from X-rays took the name of \textit{Discrete Tomography} \cite{Gardner,Herman1,Herman2}. Due to the technical system providing the measurements  and the complexity of the considered problems, a special attention is given on the problem in dimension $2$.

\subsection{Problem statement}

An \textit{X-ray} is the sequence of the number of points of the intersection between a given lattice set and the consecutive diophantine lines in a chosen direction. In the particular case of the  horizontal and vertical directions in dimension $2$, we define the horizontal and vertical X-rays of a lattice set $S \subset [0\cdots m-1]\times [0\cdots n-1]$ as the two vectors $H(S)$ and $V(S)$ counting the numbers of points of $S$ in each row and column. Their coordinates are $h_j(S)=|\{ (x,y)\in S | y=j\} |$ for $0\leq j \leq n-1$ and $v_i(S)=|\{ (x,y)\in S | x=i\} |$ for $0\leq i \leq m-1$
 (Fig.~\ref{fig2}).

\begin{figure}[!ht]
\vspace*{-3mm}
  \begin{center}
		\includegraphics[width=0.61\textwidth]{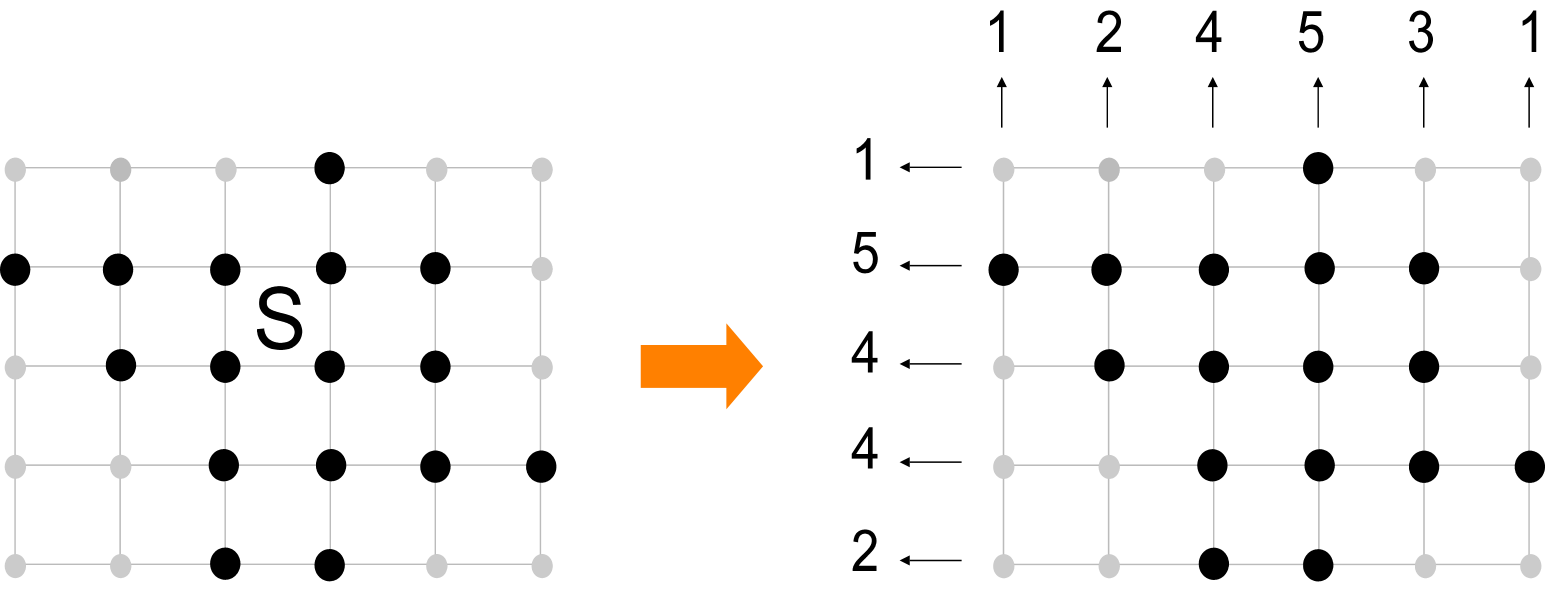}\vspace*{-5.5mm}
	\end{center}
	\caption{\label{fig2}\textbf{The horizontal and vertical X-rays} of the lattice set $S$ are the vectors $V(S)=(1,2,4,5,3,1)$ and $H(S)=(2,4,4,5,1)$.
}\vspace*{-3mm}
\end{figure}

It leads to introduce a generic problem of Discrete Tomography. The question is the existence of a lattice set with given X-rays and belonging to a given class $\A$ of lattice sets. We consider in this paper the problems of reconstruction from one or two X-rays.

\begin{pb}[$DT_\A (v)$]
Given a class $\A$ of finite lattice sets,
\newline
\textbf{Input:} one vector $v\in \Z ^m$.
\newline
\textbf{Output:} does there exist a lattice set $S\in \A$ included in the strip $[0\cdots m-1]\times \Z$ with $V(S)=v$ ?
\end{pb}

\begin{figure}[!b]
  \begin{center}
		\includegraphics[width=0.77\textwidth]{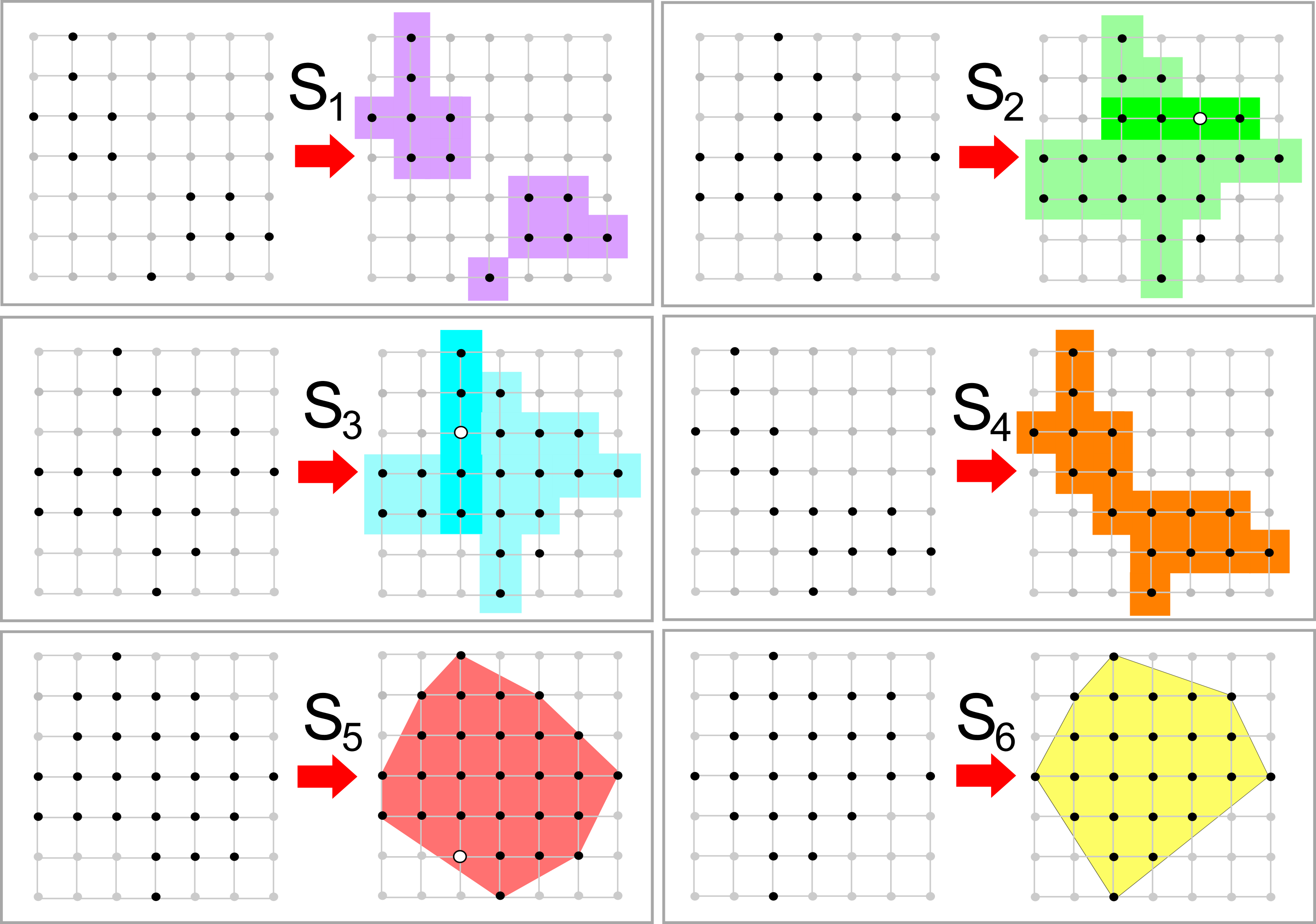}
	\end{center}\vspace*{-3.5mm}
	\caption{\label{convex}\textbf{Main classes of convex lattice sets.} The set $S_1$ is not a polyomino since not connected. The set $S_2$ is not horizontally convex,  $S_3$ is not vertically convex. The set $S_4$ is connected, horizontally and vertically convex. It is an HV-convex polyomino and then belongs to the class $\H \cap \V \cap \P$. The set $S_5$ is not digital convex while $S_6$ is digital convex ($S_6 \in \C$).
}\vspace*{-1mm}
\end{figure}

With a vertical and a horizontal X-ray instead of just a vertical X-ray, we have the following problem:

\begin{pb}[$DT_\A (h,v)$]
Given a class $\A$ of finite lattice sets,
\newline
\textbf{Input:} two vectors $h\in \Z ^n$ and $v\in \Z ^m$.
\newline
\textbf{Output:} does there exist a lattice set $S\in \A$ included in the rectangle $[0\cdots m-1]\times [0\cdots n-1]$ with $V(S)=v$ and $H(S)=h$ ?
\end{pb}

The two problems $DT_\A (v)$ and $DT_\A (h,v)$ are highly dependant on the chosen class
$\A$. In the paper, we focus our attention on convex lattice sets but different definitions might be considered (Fig.~\ref{convex}).

\begin{definition}
A lattice set $S\subset \Z ^2$ is a polyomino if it is $4$-connected namely if any pair of points of $S$ is connected by a sequence of points of $S$ at Euclidean distance $1$. The class of  polyominoes of $\Z^2$ is denoted $\P$.

A lattice set $S\subset \Z ^2$ is horizontally (vertically) convex if its intersection with any horizontal (vertical) line is a set of consecutive points. The class of horizontally (vertically)  convex lattice sets is denoted $\H$ ($\V$ for vertically convex sets). Horizontally and vertically convex sets are simply said HV-convex. Their class is $\H \cap \V$.

A lattice set $S\subset \Z^2$ is \textit{digital convex} if it is equal to its intersection with its real convex hull $S=\conv (S) $ where the operator $\conv (.)$ denotes the integer convex hull namely $\convv (.) \cap \Z ^2$. The class of the digital convex lattice sets is denoted $\C$.
\end{definition}

The class of HV-convex polyominoes is the intersection $\H\cap \V \cap \P$. It is not convex in an ordinary meaning. The digital convex sets are the intersections of convex polygons with the lattice $\Z ^2$.
Notice that the digital convex lattice sets are obviously HV-convex ($\C \subset \H \cap \V$).


\subsection{State of the art}
It is known for a long time that for the whole class denoted $\W$ of all lattice sets, the problem $DT_\W (h,v)$ can be solved in polynomial time by specific algorithms \cite{Gale,Ryser} or by min cut-max flow algorithms (Fig.~\ref{fig-0}).

\begin{figure}[ht!]
  \begin{center}
		\includegraphics[width=0.45\textwidth]{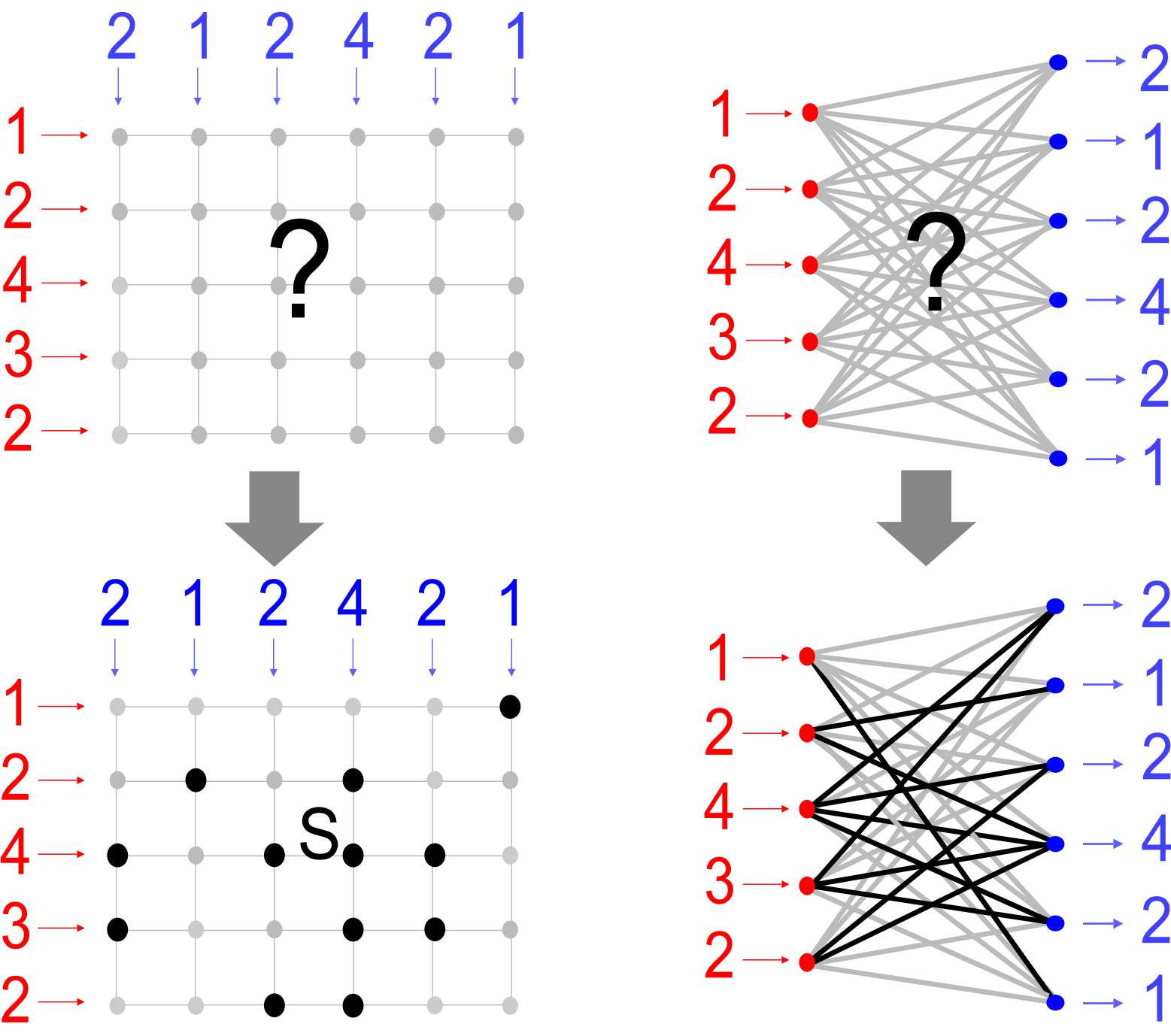}
	\end{center}\vspace*{-3mm}
	\caption{\label{fig-0} \textbf{Reduction of an instance of $DT_\W (h,v)$ to a problem of flow}. 
}
\end{figure}

Many variants of the problem $DT_\W (h,v)$ have been investigated, not only with horizontal and vertical X-rays but in different dimensions, with different  directions of X-rays and different kinds of atoms.
If we focus on the two-dimensional case with horizontal and vertical X-rays, the complexities of the problem  $DT_\A (h,v)$ have already been determined for several classes $\A$.
The problem $DT_\A (h,v)$ is NP-complete:
\begin{itemize}
    \item for the class $\A=\P$ of polyominoes \cite{W1}
    \item for the classes $\A=\H$ or $\A=\V$ of horizontally or vertically convex lattice sets \cite{Nivat}.
    \item for the class $\A=\H \cap \V$ of  HV-convex lattice sets (it is a particular case of puzzle games called \textit{nonograms}) \cite{W1}.
\end{itemize}
These $NP$-hardness results are however counter-balanced by two major results of the field published in two seminal papers:
\begin{itemize}
    \item  The problem $DT_{\H \cap \V \cap \P} (h,v)$ i.e the reconstruction of HV-convex polyominoes can be solved in  polynomial time \cite{Nivat}.
    \item Results of uniqueness have been obtained for the class $\C$ of digital convex lattice sets with different number of directions of X-rays. R. Gardner and P. Gritzmann characterized the sets of $d$ directions for which any digital convex lattice set is uniquely determined by its X-rays  \cite{GG}. 
    With the directions of X-rays providing uniqueness, these results have been completed by a polynomial time algorithm of reconstruction \cite{Bru2}. This algorithm follows the two steps used for the reconstruction of HV-convex polyominoes \cite{Nivat}.
\end{itemize}

\begin{table}[h]
\centering
\caption{\label{tab1}Milestones results and neighboring open questions}
\begin{tabular}{| p{4.3cm} | p{3.5cm} | p{3.5cm} |}
\hline
 \textbf{Class $\A$} & \textbf{Horizontal + vertical  X-rays} & \textbf{4 directions \newline or more}\\
\hline
\textbf{$\H \cap \V \cap \P$\newline (HV-convex polyominoes)} &  $DT_{\H \cap \V \cap \P} (h,v)$\newline   polynomial time  \cite{Barcucci} & open \newline\\
\hline
\textbf{$\C$ \newline (Digital convex sets)} & $DT_\C (h,v)$  \newline open  &  polynomial time \newline(if uniqueness)  \cite{GG,Bru2}\\
\hline
\end{tabular}
\end{table}

\subsection{Three results}

Our goal is to address the problem $DT_\C (h,v)$. It is a twenty years old open question of Discrete Tomography \cite{Dulio4, Niccolo}. Unfortunately, we solve it only partially by providing a polynomial time algorithm for a subclass of the digital convex sets $\C$ that we call \textit{fat}.
We present two intermediary results which have their own interest. Then the article contains three results.


\paragraph{1 - The Bad Guy.}
The first result is about the reconstruction $DT_{\H \cap \V \cap \P} (h,v)$ of HV-convex polyminoes.
It has been observed since many years that the second step of the original algorithm presented in  \cite{Nivat} always provides a solution. It became an oral conjecture. We disprove it with a first counter-example.
This counter-example has been called by Andrea Frosini the \textit{bad guy}.

\paragraph{2 - Only One X-ray.}
The second result states the complexity of the problem of reconstruction of digital convex sets from only one X-ray:

\begin{theorem} \label{t0}
The algorithm \texttt{DAGTomo1} solves $DT_{\C} (v)$ with a worst case time complexity in $O(m^{11} +  m (\sum _{i=0}^{m-1} v_i)^5 )$ where $m$ is the number of columns of the X-ray and $ \sum _{i=0}^{m-1} v_i$ the number of points to determine.
\end{theorem}

The result is interesting by itself but also by the strategy of the algorithm. It reduces the computation of a solution to the research of a path in Directed Acyclic Graph of polynomial size.

\paragraph{3 - With Two X-rays and a Fatness Property.}
The third result is about the reconstruction of digital convex lattice sets from horizontal and vertical X-rays namely $DT_{\C} (h,v)$. We provide a polynomial time algorithm called \texttt{DAGTomo2} for reconstructing the fat digital convex sets  (Fig.~\ref{fat}).
The fatness of a digital convex set is related to the positions of its points with minimal and maximal $x$ and $y$ coordinates. These extreme points are the \textit{feet} of the lattice set (Fig.~\ref{feet}).

\begin{figure}[ht!]
\vspace*{1mm}
  \begin{center}
		\includegraphics[width=0.55\textwidth]{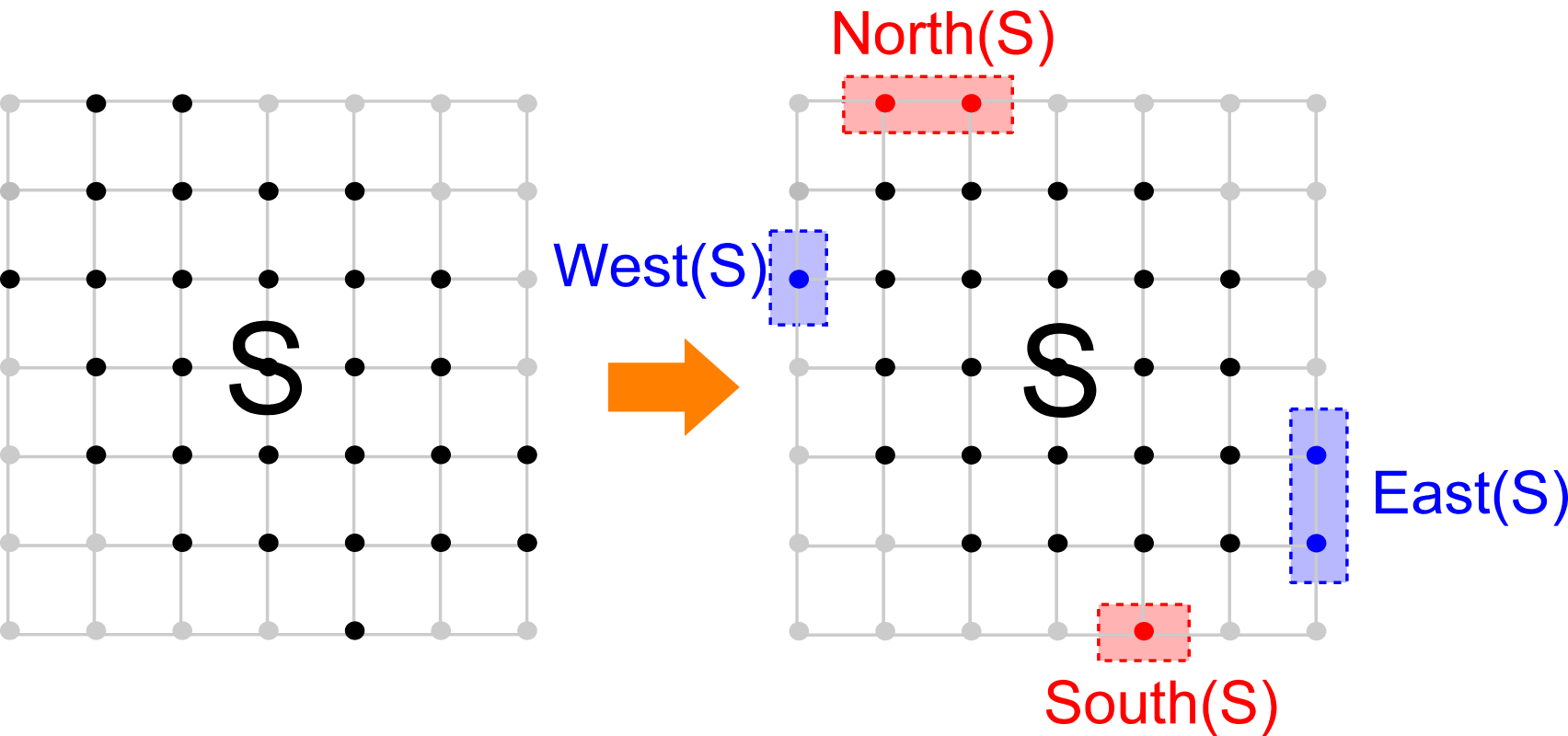}
	\end{center}\vspace*{-5mm}
	\caption{\label{feet} \textbf{The four feet} of a lattice set $S$ are the subsets denoted $\South$, $\West$, $\North$ and $\East$.
}
\end{figure}

\begin{definition}\label{regular}
Given a lattice set $S \subset \Z^2$, the South, West, North and East feet are \newline
$\South(S)=\{ (x,y)\in S | \forall (x',y')\in S, y'\geq y  \}$, \newline
$\West(S)=\{ (x,y)\in S | \forall (x',y')\in S, x'\geq x  \}$, \newline
$\North(S)=\{ (x,y)\in S | \forall (x',y')\in S, y'\leq y  \}$, \newline
$\East(S)=\{ (x,y)\in S | \forall (x',y')\in S, x'\leq x  \}$.\newline
The lattice set $S$ is thin if there exists $(X,Y)\in \Z^2$ such that the feet are strictly located in diagonally opposite quadrants of $(X,Y)$ i.e verifying\newline
either $x(\South(S)) < X < x(\North(S))$ and $y(\West(S)) < Y < y(\East(S))$,
\newline or $x(\South(S)) > X > x(\North(S))$ and $y(\West(S)) > Y > y(\East(S))$.\newline
If $S$ is not thin, it is fat (Fig.~\ref{fat}).
\end{definition}

\begin{figure}[ht!]
\vspace*{3mm}
  \begin{center}
		\includegraphics[width=0.94\textwidth]{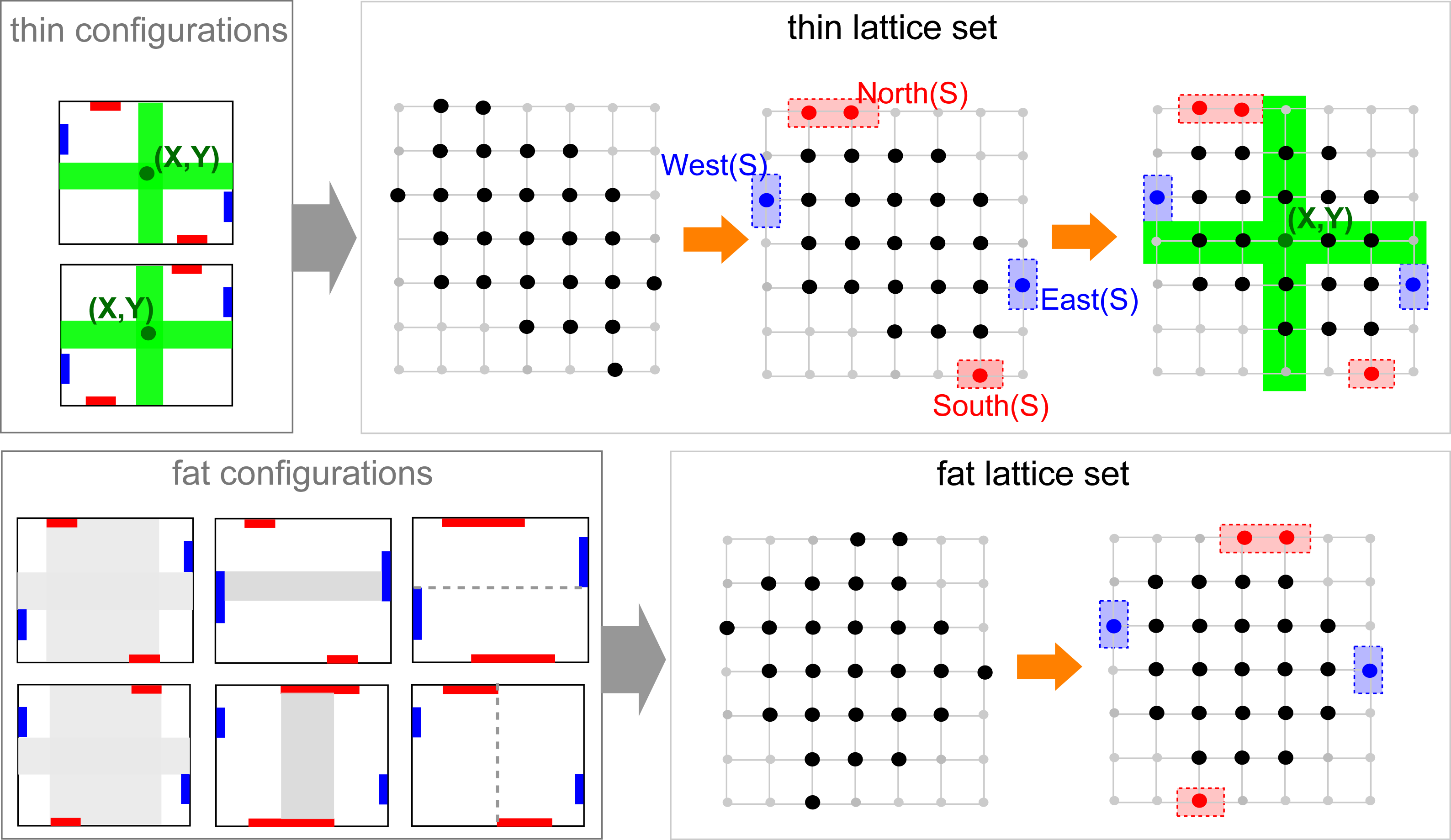}
	\end{center}\vspace*{-2mm}
	\caption{\label{fat} \textbf{Thin VS fat lattice sets.} The fatness/thinness of an HV-convex lattice set depends on the relative positions of its feet. A lattice set is thin if  there exists an integer point $(X,Y)$ (represented by the green cross) strictly separating the pairs of feet in diagonally opposite quadrants. Otherwise it is fat.
}\vspace*{-3mm}
\end{figure}

The class of the fat lattice sets (Fig.~\ref{fat}) is denoted $\F$. Then the class of the fat digital convex sets is $\C \cap \F$.
The last result of the paper is the following  theorem:

\begin{theorem}\label{t1}
The algorithm \texttt{DAGTomo2} solves $DT_{\C \cap \F} (h,v)$ with a worst case time complexity in $O(m^7n^7)$.
\end{theorem}
The time complexity of \texttt{DAGTomo2} is high but polynomial. The complexity of the reconstruction of the thin  digital convex sets $DT_{\C \setminus \F} (h,v)$ remains an open question.\\

\subsection{Plan}

The paper is organized with one section per result. The second section is focused on the reconstruction of HV-convex sets. It presents the bad guy counter-example and introduces the material necessary to prove Theorem \ref{t1}.

\medskip
The third section  proves Theorem \ref{t0} by presenting a polynomial time algorithm \texttt{DAGTomo1} for solving $DT_{\C} (v)$.
The fourth section presents the polynomial time algorithm \texttt{DAGTomo2} for solving $DT_{\C\cap \F} (h,v)$. It proves Theorem \ref{t1}.


\section{The structure of the bad guy counter-example}

This section is focused on the reconstruction of HV-convex polyominos. It summarizes the results of the two previous episodes \cite{Nivat} and \cite{RSC} by defining the \textit{switching components} and provides a counter-example of a conjecture.

\subsection{Episode 1: the reconstruction procedure with its filling/aggregation strategy}

The milestone paper \cite{Nivat} proved that the problem $DT_{\H \cap V \cap \P} (h,v)$ can be solved in polynomial time. There exists an extensive literature around this problem \cite{Dulio2,Dulio3} and several algorithms have been derived from this algorithm \cite{Dulio1,Bru2,Frosini10,Sara}.
We recall its two steps strategy and how it deals with the ambiguities of the reconstruction.

The strategy of the algorithm is to assume that some points of the solution are known and then to perform a reconstruction procedure. The given subset of a solution $S$ is denoted $\In _0$.
The algorithm tries several subsets $\In _0$ and for each, the algorithm repeats the reconstruction procedure.

\paragraph{Choosing the Feet.} The first set of points $\In _0$ assumed to belong to the solution is not chosen randomly. The original strategy of \cite{Nivat} is to choose as set $\In _0$ the four feet $\North$, $\South$, $\East$, and $\West$  (Fig.~\ref{feet}).
The $\North$ and $\South$ feet might have at most $m$ different positions while the $\East$ and $\West$ feet might have at most $n$ different positions. It leads us to perform the reconstruction procedure for at most $m^2n^2$ different initial kernels $\In _0$ (Fig.~\ref{fig-feet}).

\begin{figure}[ht]
\vspace*{1mm}
  \begin{center}
		\includegraphics[width=0.95\textwidth]{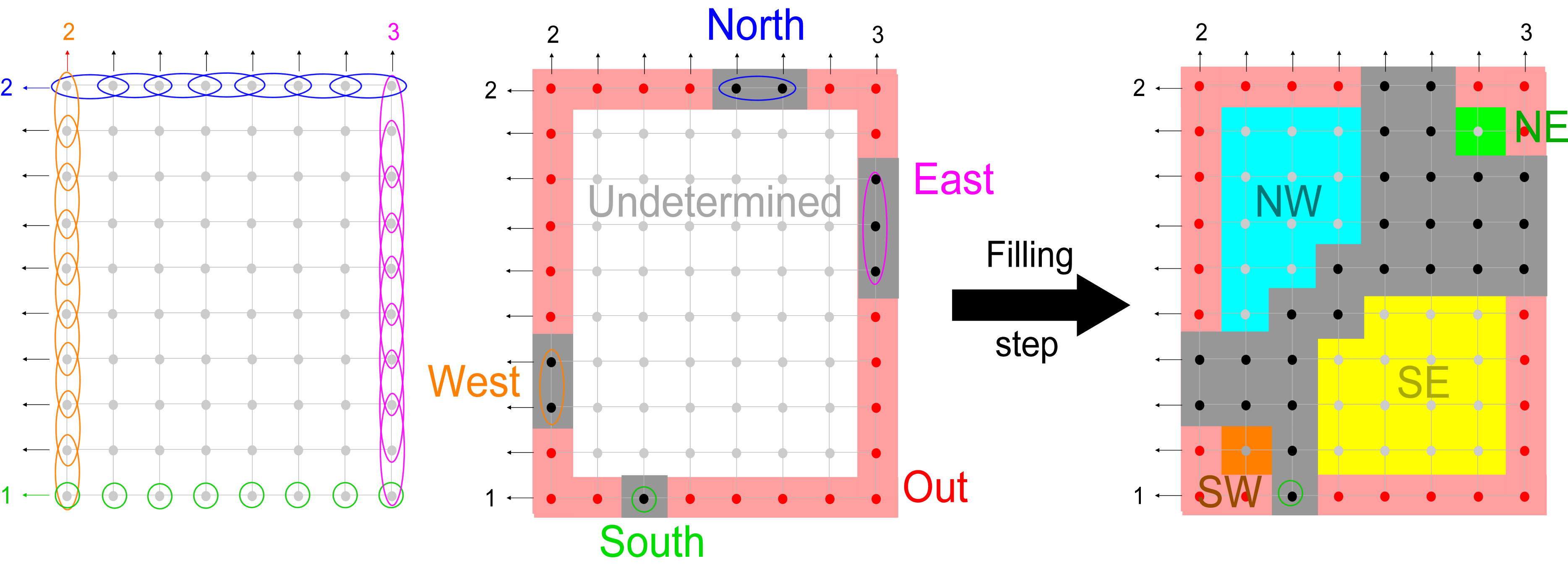}
	\end{center}\vspace*{-5mm}
	\caption{\label{fig-feet} \textbf{The feet and the corners}. The kernel is initialized by choosing the positions of the four feet $\North$, $\South$, $\East$, and $\West$. At each stage, the undetermined points are represented by small grey disks. The kernel points are represented by small black disks in grey cells. The excluded points are  red disks in pink cells. After the filling step, the set of the undetermined points is $4$-connected. It provides a partition of the undetermined points in $4$ corners.
}\vspace*{-1mm}
\end{figure}

We present now the reconstruction procedure. It uses a partition of the grid in three sets presented in the following paragraph. The two steps of the reconstruction procedure i.e the filling and aggregation steps are presented in the later.

\paragraph{Kernel, Excluded and Undetermined Points.}

During the reconstruction procedure, the current knowledge of the solution is described by a partition of the rectangular region of interest $[0\cdots m-1]\times [0\cdots n-1]$ in three sets denoted $\In$, $\Out$, and $\Undetermined$.
\begin{itemize}
    \item The kernel $\In$ contains the integer points which are currently known to belong to all solutions (up to the considered subset $\In _0$).
    \item The set $\Out$ contains the integer points which are currently known to be excluded from all solutions.
    \item The set of the integer points whose status \textit{In or Out of a solution} has not yet  been determined is denoted $\Undetermined$. It is also often called the \textit{shell}.
\end{itemize}

For each execution of the reconstruction procedure, the kernel $\In$,  the set of the excluded points $\Out$ and the set $\Undetermined$ of the undetermined points are respectively initialized with $\In \leftarrow \In _0$, $\Out\leftarrow \emptyset$ and $\Undetermined \leftarrow [0\cdots m-1]\times [0\cdots n-1] \setminus \In _0$.

\paragraph{Filling Step.}

Once the three sets $\In$, $\Out$ and $\Undetermined$ have been initialized, the prescribed X-rays and the prescribed properties of the solution ($4$-connectivity and HV-convexity) can be used to decrease the set of the undetermined points. Some of them are added to the current solution $\In$ while some other points cannot belong to the solution and are added to $\Out$. Then the new determined points allow us to determine new points. The algorithm repeats this process until no undetermined point could be determined. This step of the reconstruction procedure from the initialization to the stage where no undetermined point can be determined without ambiguity is called the \textit{filling step}.

We refer to   \cite{Bru,Bru2} for a  complete presentation of the filling operations and suitable data structures.

\begin{figure}[ht]
  \begin{center}
		\includegraphics[width=0.98\textwidth]{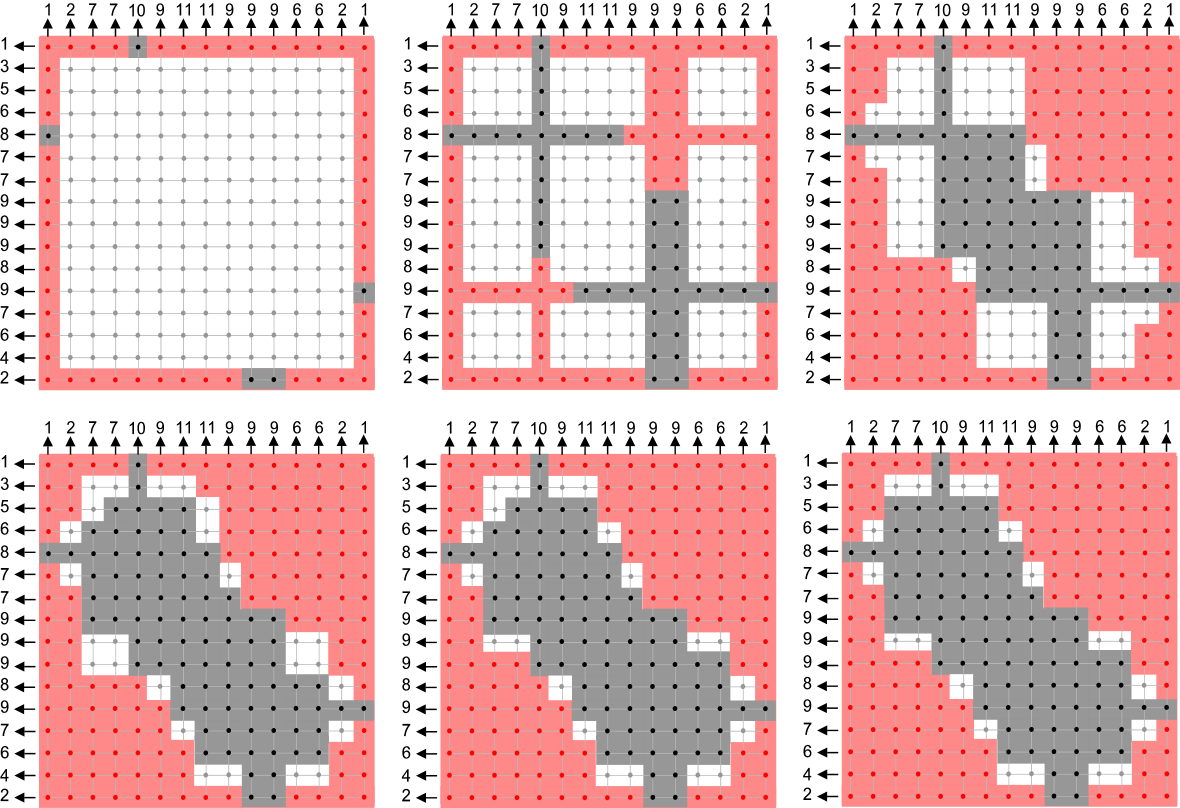}
	\end{center}\vspace*{-2mm}
	\caption{\label{fig-6} \textbf{Filling step.} On the bad guy instance starting from a given position of the feet, we show the partition of the region of interest in $\In \cap \Out \cap \Undetermined$ at different stages. 
	}\vspace*{-1mm}
\end{figure}

The filling operations run until either finding a contradiction (a point of $\In$ that should be also assigned to $\Out$ or conversely) or until no undetermined point could be determined without ambiguity.

We illustrate the filling step on the \textit{bad guy} instance
$DT_{\H \cap \V \cap \P} (h_0,v_0)$ with
\begin{equation}\label{bad1}
    h_0=(2,4,6,7,9,8,9,9,9,7,7,8,6,5,3,1)
\end{equation} and
\begin{equation}\label{bad2}
v_0=(1,2,7,7,10,9,11,11,9,9,9,6,6,2,1).
\end{equation}
The chosen  position of the feet is  $\South = \{ (10,0), (11,0) \}$, $\North = \{ (5,15) \}$,
$\East = \{(0,12) \}$ and $\West = \{(14,5) \}$ in Fig.~\ref{fig-6}.

\medskip
If the filling operations find a contradiction, there is no solution with the considered position of the feet. This branch of the computation is stopped.
If there is no contradiction and it does not remain any undetermined point, then a unique solution has been found with the considered position of the feet. This branch of the computation also ends.
We go to the aggregation step only if it remains undetermined points without contradiction.

\paragraph{The Switching Components.}

We place us now in the case where the filling step did not find any contradiction and where it remains some undetermined points. Then the kernel obtained by using the filling operations of \cite{Nivat} has the properties to be $HV$-convex and $4$-connected.
The connectivity allows us to partition
 the undetermined points in four subsets that we call the \textit{borders}: the North West, North East, South East and South West borders. They are respectively denoted $\NW$, $\NE$, $\SE$ and $\SW$ (Fig.~\ref{fig-feet}).

Then we have a partition of the undetermined points in $\Undetermined= \NW \cup \NE \cap \SE \cap \SW$. We can build another partition in the following way.
The starting remark is that an undetermined point cannot be isolated in a row or a column. Otherwise it would have been determined. The undetermined points go by pair on each row and column.
If for instance $\underline p=(i,j)$ is undetermined and belongs to the South border $\SW \cup \SE$, then the point denoted  $\overline p  =(i, j+v_i)$ is a North undetermined point. These two points are \textit{vertical correspondents}. In the same way, any West undetermined point $|p=(i,j)$ has a horizontal East correspondent
$p|=(i+h_j,j)$ (Fig.~\ref{fig8}).
Horizontal and vertical correspondences are symmetric relations.

\begin{figure}[ht]
  \begin{center}
		\includegraphics[width=0.42\textwidth]{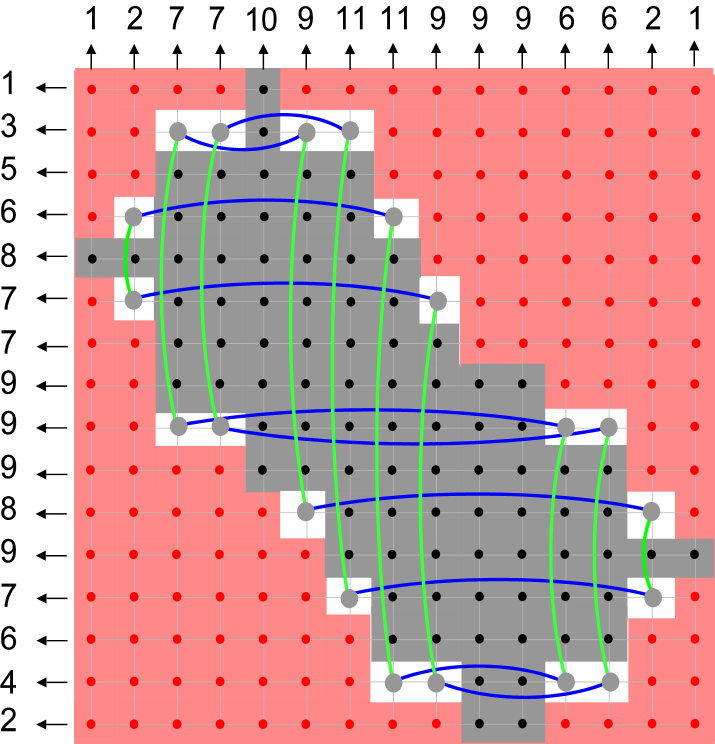}
	\end{center}\vspace*{-3mm}
	\caption{\label{fig8} \textbf{Corresponding points and switching components}. After the filling step, horizontal correspondences are drawn in blue while the vertical correspondences are drawn in green on the bad guy instance. In this case, the graph is composed of a unique cycle i.e switching component.
}
\end{figure}

\begin{figure}[!h]
\vspace*{2mm}
  \begin{center}
		\includegraphics[width=0.98\textwidth]{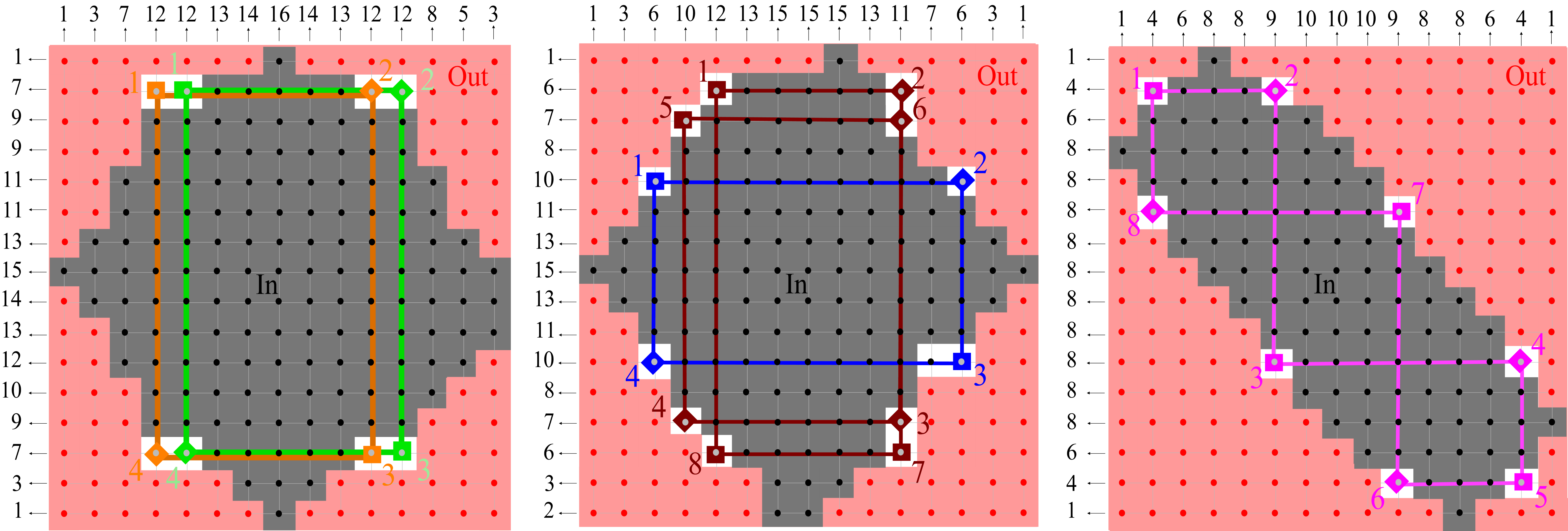}
	\end{center}\vspace*{-3mm}
	\caption{\label{fig9} \textbf{Switching components}. In each switching component, we represent the points alternatively by squares $\square$ and diamonds $\diamond$ to express that their Boolean variables are negation from each other. Either the squares, or the diamonds of a switching component are in a solution.}\vspace*{-3mm}
\end{figure}

\medskip
We present now the main property of corresponding points (Fig.~\ref{fig9}). If no one of two corresponding points is aggregated to the kernel then the number of points on its row or column is too small. Conversely, if both points are aggregated to the kernel, there are too many points on the row or column. It follows that exactly one point per pair of corresponding points has to be aggregated to the kernel.
We associate now to each undetermined point a Boolean variable assigned to $1$ if we decide to aggregate the point to the kernel and to $0$ if it is rejected. These variables are called the \textit{aggregation variables}.
It follows from the previous remark that the values of the aggregation variables of corresponding points are the negation from each other. We represent graphically this relation by denoting the points either with a square $\square$ or with a diamond $\bigDiamond$.

If we consider now the graph of the corresponding points, each point is a vertex of degree $2$. It follows that the graph of correspondences is an union of cycles each one called \textit{switching component}.

\begin{definition}
A switching component $P$ is a cycle of corresponding undetermined points (Fig.~\ref{fig9}).
\end{definition}

\paragraph{Aggregation Step.}

It remains us to determine which undetermined points should to be aggregated the kernel or rejected. In other words, the problem is to find an assignation of the aggregation variables associated to the undetermined points which provides a set which is $4$-connected, HV-convex and has the prescribed numbers of points in each row and column. The key point of the algorithm from \cite{Nivat} is that these properties are easily expressed by a 2-CNF on the aggregation variables. Then finding an HV-convex polyomino with the prescribed horizontal and vertical X-rays is reduced to a 2-SAT instance. Then the aggregation step is resolved   by computing and solving the 2-CNF in polynomial time \cite{Tarjan}.

\paragraph{Old Oral Conjecture.}
25 years after the seminal paper presenting this algorithm \cite{Nivat}, no instance was found where the 2-SAT CNF on the aggregation variables was not feasible. This observation led to the following empirical conjecture: the 2-SAT formulas expressing the HV-convexity of a solution on the aggregation variables always feasible.

To determine whether this conjecture is true or false and also for the reminder of the paper, we now present some structural properties of the switching components.

\subsection{Episode 2:  structures of the switching components}

Before going further, notice that the properties used for the reconstruction procedure of an HV-convex solution
can also be used for the reconstruction of digital convex lattice sets since digital convexity implies HV-convexity.

\medskip
According to \cite{RSC}, the structures of the switching components can be classified according to the positions of the feet. There are six different configurations presented in Fig.~\ref{abcdef}. In the configurations a), b), c), d), e), the positions of the feet provide thin solutions while the configuration f) is the configuration of the fat convex lattice sets.

\begin{figure}[h!]
\vspace*{2mm}
  \begin{center}
		\includegraphics[width=0.96\textwidth]{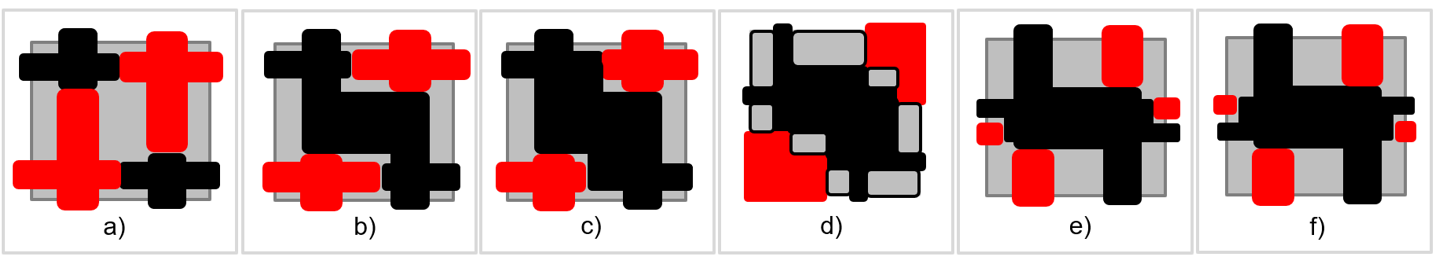}
	\end{center}\vspace*{-4mm}
	\caption{\label{abcdef} \textbf{The  six feasible configurations of the feet} admitting HV-convex solutions \cite{RSC}. We color in black the feet and the points that can be directly determined because they are in the same rows or columns. We color in red some of the points which can be directly excluded and in grey the  regions of the undetermined points.
}
\end{figure}

\begin{figure}[!htbp]
\vspace*{-3mm}
  \begin{center}
		\includegraphics[width=0.74\textwidth]{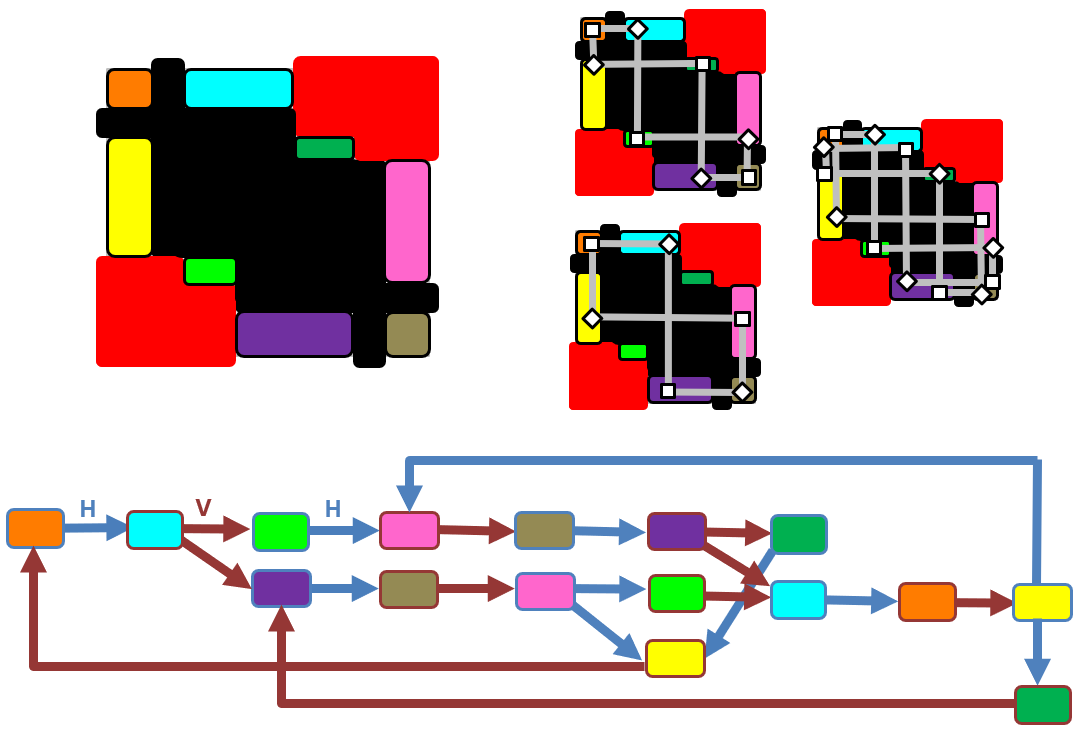}
	\end{center}\vspace*{-4mm}
	\caption{\label{paths} \textbf{In the configuration c) the graph of correspondences  between the undetermined regions}. We consider the graph of vertical and horizontal correspondences of the different regions that might contain undetermined points (with different colors). Each region appears twice in the graph, once after a horizontal correspondence and the other after a vertical correspondence.
	This graph admits six fundamental oriented cycles and three without orientation. With feet in configuration c), the switching components are necessarily concatenations of these three fundamental cycles.
}
\vspace*{3mm}
  \begin{center}
		\includegraphics[width=0.89\textwidth]{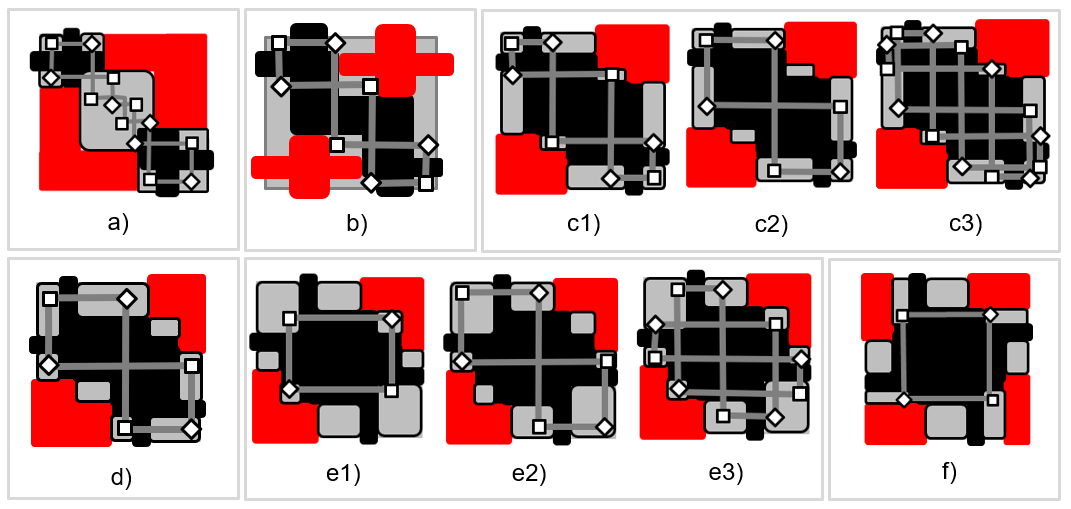}
	\end{center}\vspace*{-4mm}
	\caption{\label{cycles} \textbf{The fundamental cycles of the switching components} in each configuration (except the configuration a) for which the path can have an arbitrary number of edges in the central region). Notice that the cycle c3) has been forgotten in \cite{RSC}. The list is now complete.
}
\end{figure}

For each one of the six configurations, the set of undetermined points is restricted to a small number of regions. Each region corresponds horizontally and vertically to some other regions, so that we can build the graph of the vertical and horizontal correspondences of the regions. The correspondence graph in the configuration c) is for instance represented in Fig.~\ref{paths}. Such graphs admit a small number of fundamental cycles i.e minimal closed paths. There are exactly three minimal cycles in the configuration c). The closed path in the graph are obtained by concatenation of these minimal cycles. The
 fundamental minimal cycles in all configurations are represented in  Fig.~\ref{cycles}.
They provide a better understanding of the very constrained structures of the switching components.

\subsection{The bad guy}

Let us come back to the conjecture that the aggregation 2-SAT formula expressing the HV-convexity of a solution is always feasible. For the configurations b), d), and f), it is true: a solution can be found by aggregating to the kernel either all the squares or all the diamonds. For having a chance to provide an example of a non feasible aggregation formula, we need squares and diamonds in the same region as in the following counter-example.

\medskip
The bad guy instance is in configuration c) and we have a unique switching component which is a concatenation of a fundamental cycle c2) and a fundamental cycle c3) (these two cycles are drawn in Fig.~\ref{cycles}). The filling step is illustrated in Fig.~\ref{fig-6}. The correspondences are illustrated in Fig.~\ref{fig8} and the Boolean aggregation variables are represented in Fig.~\ref{badguy}.
The 2-SAT formula expressing the HV-convexity of a solution admits no solution since a square implies a diamond, which is contradictory. It disproves the empirical conjecture. There exists reconstruction instances for which the aggregation 2-SAT CNF is not feasible.

\begin{figure}[h!]
\vspace*{2mm}
  \begin{center}
		\includegraphics[width=0.74\textwidth]{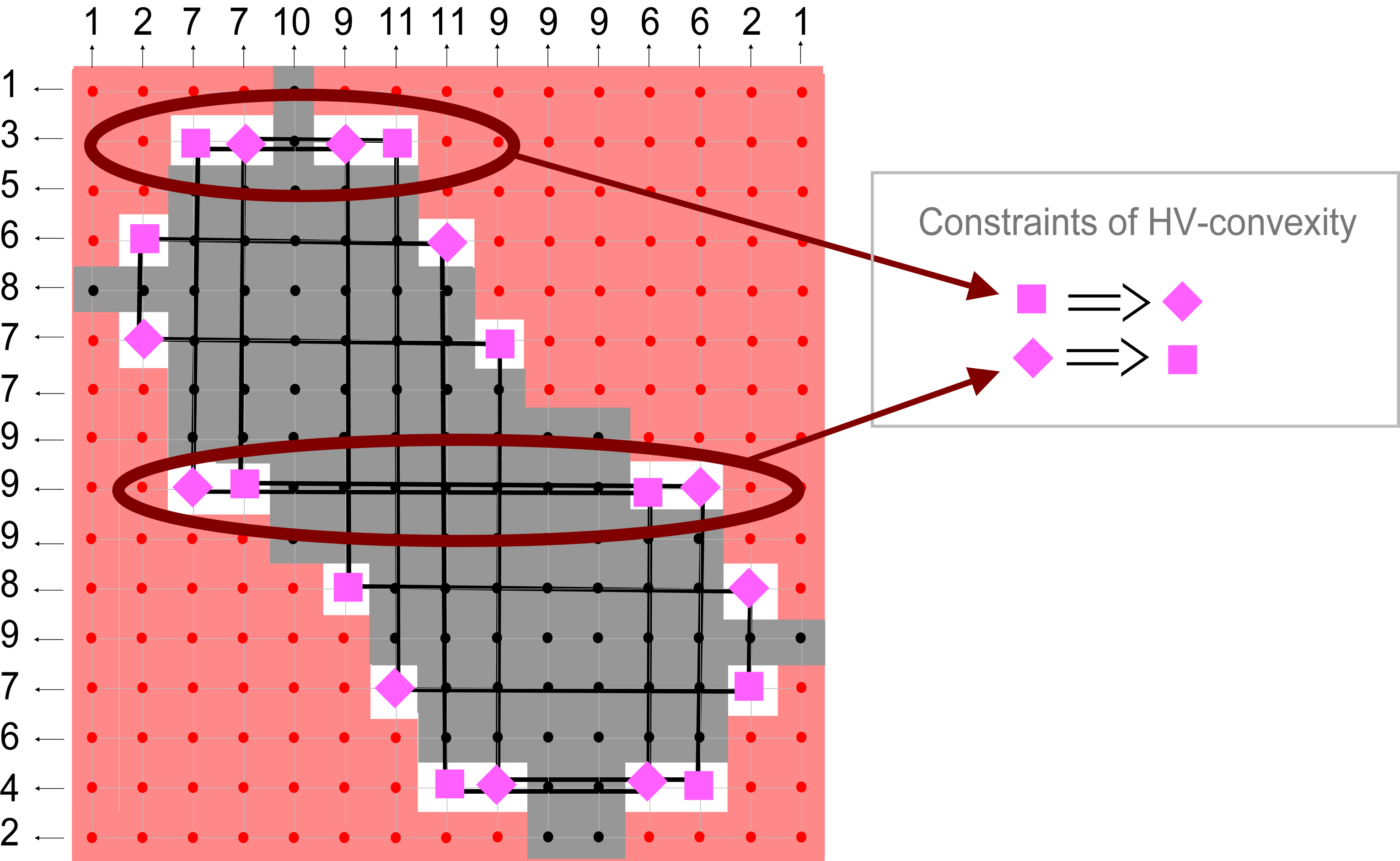}
	\end{center}
	\caption{\label{badguy} \textbf{Counter-example of the conjecture.} The filling operations lead to a unique switching component. It was conjectured that when the filling operations do not lead to an inconsistency, there exists always an assignment of the switching components providing an HV-convex solution. It is false with this example. The HV-convexity leads to inconsistent 2-clauses.
}
\end{figure}

\section{Reconstructing digital convex sets from a unique X-ray in polynomial time}\label{g}

We now consider the problem $DT_{\C} (v)$ of the reconstruction of a digital convex lattice set from a unique vertical X-ray. It is a fundamental question which provides new tools for tackling the problem with horizontal and vertical X-rays.

In this section we provide a polynomial time algorithm called \textit{DAGtomo1} for solving $DT_{\C} (v)$. We start by reducing the problem to a bounded region of interest. Then we provide a reduction of the problem to the research of a path in a DAG.

\subsection{Work in a bounded region}

We assume without loss of generality  that the first and last coordinates of the input vector $v$ are not null. It guarantees that any solution $S$ of $DT_\C (v)$ contains points in the vertical lines $x=0$ and $x=m-1$.

\medskip
The lowest points of a solution $S$ on these two lines are respectively called the left and right bottom points.
Then instead of searching an arbitrary solution of $DT_\C (v)$, we are going to restrict our research to a solution where the left and right bottom points belong to a bounded region.

The first remark is that vertical translations preserve the vertical X-ray and the digital convexity. Then the set of the solutions of $DT_\C (v)$ is invariant under vertical translations. It follows that if the instance $DT_\C (v)$ admits a solution, then there is a solution with the origin as left bottom point.
It allows us to reduce the problem $DT_\C (v)$ to the research of a solution with the origin as left bottom point.

The second remark is that the vertical shearing $shear: x\leftarrow x , y \leftarrow y-x$ preserves the vertical X-ray, the convexity, and the vertical line $x=0$.
Then the set of the solutions  of the instance $DT_\C (v)$ is invariant by the vertical shear transform. It follows that if $DT_\C (v)$ admits a solution $S$, then it admits a solution $S'$  with the same left bottom point than $S$ and such that the $y$ coordinate of the right bottom  point of $S'$ is in the interval $[0, m-1[$ (we recall here that the $x$ coordinate of the right bottom   points is $m-1)$.

With these two remarks, we restrict the research of a solution of $DT_\C (v)$ to a digital convex lattice set with the origin as left bottom point and a right bottom point of coordinates $(m-1,y_{m-1})$ with $0\leq y_{m-1} < m-1$.

It follows that the points of the segment joining the two extreme bottom points are included in the triangle $T_m$ of vertices $(0,0)$, $(m-1,0)$, and $(m-1,m-1)$ (Fig.~\ref{quads}). This triangle $T_m$ contains $\frac{m(m+1)}{2}$ integer points.

\begin{figure}[!h]
  \begin{center}
		\includegraphics[width=0.35\textwidth]{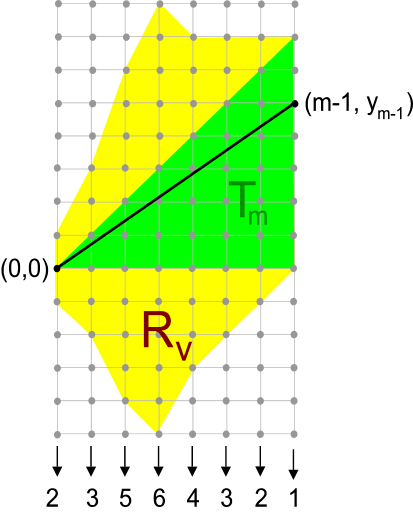}
	\end{center}\vspace*{-4mm}
	\caption{\label{region} \textbf{Region of work.} We fix the left bottom point at the origin. Then with vertical shearing, we can assume without loss of generality that the right bottom point $(m-1,y_{m-1})$ belongs to the vertical segment from $(m-1,0)$ to $(m-1,m-2)$. We notice that the segment from the origin to $(m-1,y_{m-1})$ is by definition included in the convex hull of the solution and in the triangle $T_m$ of vertices $(0,0)$, $(m-1,0)$, $(m-1,m-1)$. By taking into account  the vertical X-ray $v$, it provides a bounded region $R_v$ containing a solution if there exists one.
}
\end{figure}

  \begin{figure}[!h]
  \begin{center}
		\includegraphics[width=0.95\textwidth]{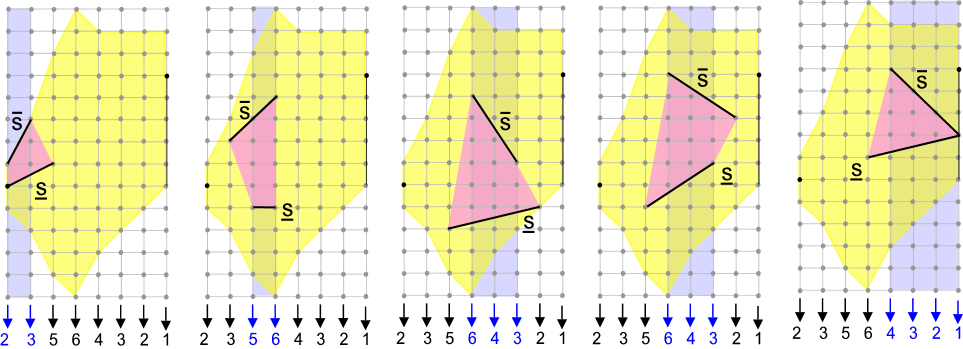}
	\end{center}\vspace*{-4mm}
	\caption{\label{quads} \textbf{Quads.} The nodes of the DAG are quads (colored in pink). The common columns of $\underline s $ and $\overline s$ are highlighted in blue. For each one of these columns, the number of points of the quad is equal to the prescribed value given by the vertical X-ray.
}\vspace*{-1mm}
\end{figure}

\medskip
We bound now the region of the other points of a solution. We consider the column $x=i$.
The coordinate of the vertical X-ray along the line $x=i$ is $v_i$. We notice that if a point $(i,y)$ belongs to a solution of $DT_\C (v)$, then its vertical distance to the triangle $T_m$ is at most $v_i -1$. Otherwise the solution would not be convex. The region whose points $(i,y)$ are at vertical distance at most $v_i$ from $T_m$ is denoted $R_v$.
It contains at most $\frac{m(m+1)}{2} + 2 \sum _{i=0}^{m-1} v_i -2m $ points where $\sum _{i=0}^{m-1} v_i$ is the prescribed number of points of any solution. We have shown that if $DT_\C (v)$ admits a solution, then there exists a solution included in $R_v$. The algorithm \texttt{DAGTomo1} searches for a solution in this region.

\subsection{\texttt{DAGTomo1}}

We present now the algorithm \textit{DAGtomo1} which solves $DT_{\C} (v)$ by searching for a solution in the region $R_v$. The strategy is to reduce the problem to the search of a path in a directed acyclic graph (DAG).

\begin{figure}[h!]
  \begin{center}
		\includegraphics[width=0.95\textwidth]{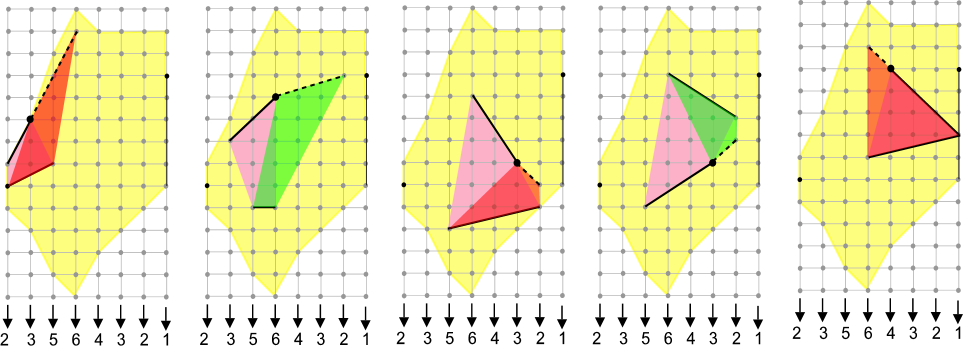}
	\end{center}\vspace*{-3mm}
	\caption{\label{links} \textbf{Oriented edges in the DAG.} There is an edge between a pair of consecutive quads if they share a segment and provide a strictly convex turn on the other side. The right quad is colored in red if it is not a successor of the left quad and in green if they are both linked by an edge in the DAG.
}
\end{figure}

\paragraph{The DAG.} We consider the following Directed Acyclic Graph $\Gamma(v)$ whose nodes are convex quads with the prescribed vertical X-ray and linked by edges if they have a convex junction (Fig.~\ref{quads} and Fig.~\ref{links}). To be more precise,
\begin{itemize}
    \item the nodes of $\Gamma(v)$ are quads made by a pair of non vertical segments $\underline s , \overline s$ in the following configuration.\begin{enumerate}
        \item First, the four vertices of the two segments $\underline s$ and $\overline s$ are in convex position (one vertex can be in common if it is the origin or if its $x$ coordinate is $m-1$) and belong to $R_v$.
        \item We denote $[a\cdots b]$ the interval of the common $x$-coordinates of the two segments $\underline s , \overline s$. We assume that the integral interval $[a..b]$ is not empty. Otherwise, the pair of segments  $\underline s , \overline s$  is excluded from the nodes.  
        \item The last condition is that the segment $\overline s$ is above the segment $\underline s$ with the prescribed X-ray. By denoting $\overline y(i)$ the $y$ coordinate of the point of $\overline s$ of $x$ coordinate $i$ and $\underline y(i)$ the $y$ coordinate of the point of $\underline s$ of $x$ coordinate $i$, we request $\lfloor \overline y(i) \rfloor - \lceil \underline y(i) \rceil = h_i$. In other words,  the quad made by the convex hull of $\underline s , \overline s$  must contain the prescribed number of points for each column with $x$ in $[a..b]$.
    \end{enumerate}

    \item We consider two nodes $\underline s_1 , \overline s_1$ and  $\underline s_2 , \overline s_2$. We denote  $[a_1..b_1]$ (respectively $[a_2..b_2]$) the interval of the common $x$ coordinates for the two pairs of segments $\underline s_1 , \overline s_1$   (respectively $\underline s_2 , \overline s_2$). The DAG $\Gamma(v)$ has an oriented edge from $\underline s_1 , \overline s_1$ to  $\underline s_2 , \overline s_2$ if
    \begin{enumerate}
        \item they have a common segment: either $\underline s_1 = \underline s_2$ or  $\overline s_1 = \overline s_2$,
        \item the two segments $s_1$ and $s_2$ which are not equal are consecutive: the right vertex $x$ of $s_1$ is equal to the left  vertex of $s_2$.
        \item Moreover, the two different segments are in a convex configuration: the vertex $x$ is not in the convex hull of the other vertices.
    \end{enumerate}
\end{itemize}

We introduce now two sets of nodes called $\Left$ and $\Right$ containing respectively the left-most and right-most quads.

The set $\Left$ contains the quads with the origin as vertex. If the first coordinate of the X-ray $v$ is $v_0=1$, then the quads $\underline s , \overline s$ of $\Left$ are only triangles with the origin as left vertex of $\underline s$ and $\overline s$. Otherwise the two segments $\underline s , \overline s$ of the quads in $\Left$ have respectively $(0,0)$ and $(0,v_0)$ as left nodes.

The set $\Right$ contains the quads with at least one vertex of $x$ coordinate $m-1$.  If the last coordinate of the X-ray $v$ is $v_{m-1}=1$, then the quads $\underline s , \overline s$ of $\Right$ are only triangles with the right vertex of coordinates $x=m-1$ and $0\leq y \leq m-1$. Otherwise, the quads $\underline s , \overline s$ of $\Right$ have a vertical edge between two vertices $(m-1,y_r)$ with $0\leq y_r \leq m-1$ and the point of coordinates $(m-1, y_r+v_ {m-1}-1)$.

\paragraph{Reduction.}

We reduce the reconstruction  $DT_{\C} (v)$ of a digital convex set of vertical X-ray $v$ to the research of path going from $\Left$ to $\Right$ in the DAG $\Gamma(v)$ (Fig.~\ref{path}).

\begin{figure}[!h]
\vspace*{2mm}
  \begin{center}
		\includegraphics[width=0.94\textwidth]{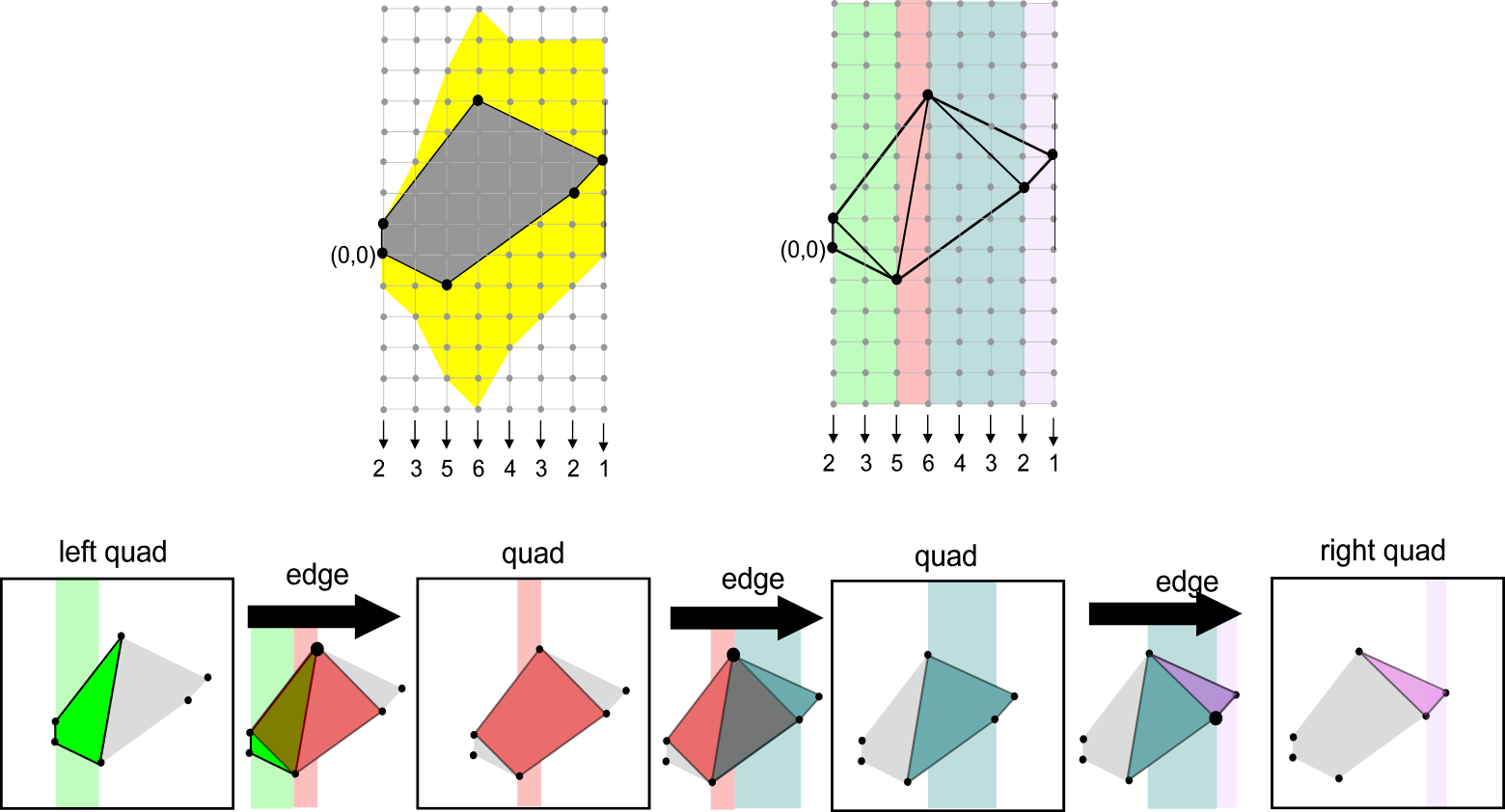}
	\end{center}\vspace*{-4mm}
	\caption{\label{path} \textbf{From a solution to a path in the DAG and conversely.} We reduce the problem of finding a solution of $DT_{\C} (v)$ to the research of a path from $\Left$ to $\Right$ in the DAG $\Gamma(v)$.
}\vspace*{-3mm}
\end{figure}

\begin{proposition}
The problem $DT_{\C} (v)$ admits a solution if and only if there exists a path going from $\Left$ to $\Right$ in the DAG  $\Gamma(v)$.
\end{proposition}

\begin{proof}
By construction, if there exists a path in the DAG $\Gamma(v)$, we consider the solution $S$ obtained by union of its quads.
The constraints on the quads to be a node guarantee the prescribed X-ray while the constraint on the direction of the turn at each edge guarantee the convexity of $S$. Thus a path from $\Left$ to $\Right$  provides a solution of $DT_{\C} (v)$.

Conversely, if we have a solution $S$ of $DT_{\C} (v)$, its convex hull is fully covered by a sequence of consecutive quads. It provides a path from $\Left$ to $\Right$. It can be noticed that if one of the quads has a vertical edge, the path is not unique since we can advance either on the upper side or on the lower side of the convex hull of $S$.
\end{proof}

\paragraph{Algorithm.}
The algorithm \texttt{DAGTomo1} builds a part of the DAG $\Gamma (v)$ and searches for a path from $\Left$ to $\Right$.

The graph $\Gamma (v)$ contains at most $k^4$ nodes/quads where $k$ is cardinality of the region $R_v$. Each node has at most $2k$ edges since one of the left nodes is replaced by another point of $R_v$.
There is no interest in computing the whole graph $\Gamma (v)$ since we search for a path issued from $\Left$. The algorithm \texttt{DAGTomo1}  computes the nodes reached from $\Left$. The main loop of the algorithm is the following: given a node/quad $q$, we consider the at most $2k$ quadrilaterals
that can be consecutive neighbors of $q$.
We test whether the new quadrilateral has already been reached in a sorted list of all the quadrilateral which have already been considered. This test takes $O(\log (k^4))$ time. If it is the first time that we meet it, we test whether its X-ray  is equal to the prescribed X-ray for the common columns of the two segments. Computing the X-ray of a quad takes at most $O(m)$ time.  If the X-ray is valid, we test
the convexity of the turn in constant time and store $q$ as antecedent. At last, we test whether this quadrilateral belongs to $\Right$ in constant time.

The algorithm \texttt{DAGTomo1} ends either when a path reaching $\Right$ is found or when all the nodes reached by starting from $\Left$ have been explored.
Its worst-case time complexity is  $O(k^4 (m + \log (k)))$. In order to provide a simple bound, we use  $m + \log (k)\leq mk$. It provides a worst-case time complexity in $O(mk^5)$. We have already bounded $k$ by $\frac{m(m+1)}{2} + 2 \sum _{i=0}^{m-1} v_i -2m $. With the inequality $(a+b)^n \leq (2a)^n + (2b)^n = 2^n (a^n + b^n)$ for positive $a$ and $b$, the property that $k$ is $O(m^2 + \sum _{i=0}^{m-1} v_i)$ implies that $k^5$ is $O(m^{10}+(\sum _{i=0}^{m-1} v_i)^5)$. It shows that the complexity of \texttt{DAGTomo1} is  $O(m^{11}+m(\sum _{i=0}^{m-1} v_i)^5)$ and proves Theorem \ref{t0}.

\section{Reconstructing fat digital convex sets in polynomial time}

In this section, we present the algorithm \texttt{DAGTomo2} for reconstructing fat digital convex sets from their horizontal and vertical X-rays. The property that this algorithm runs in polynomial time  proves Theorem \ref{t1}.

The algorithm \texttt{DAGTomo2} uses a two-step reconstruction procedure as presented for reconstructing HV-convex polyominoes with first the filling step and secondly the aggregation step.
We update the filling step for taking into account the digital convexity which replaces the connectivity and the HV-convexity. The updated filling operations are presented in Section \ref{fil}. The aggregation step used for HV-convex polyominoes is completely modified by using a DAG strategy as in Section~\ref{g}. This main new step is presented in section \ref{gla}.

\subsection{Initialization}

As the classical algorithm reconstructing HV-convex polyominoes and described in Section 2, the algorithm  \texttt{DAGTomo2} works with a partition of the grid in three sets of points $[0\cdots m-1]\times [0\cdots n-1]=\In \cup \Out \cup \Undetermined$.
The kernel is initialized by choosing a position of the feet: $\In \leftarrow \South \cup  \West \cup \North \cup \East$. Notice that we do not only compute the kernel $\In$ but we will also maintain at each step its convex hull.
For each position of the feet, we proceed to the \texttt{DAGTomo2}  reconstruction procedure. This reconstruction procedure has two steps, first the filling step and secondly the aggregation step.

\subsection{Filling step}\label{fil}

We precise now the filling operations of the reconstruction procedure. Some of the filling operations are different from the ones used for reconstructing HV-convex polyominoes because we no more assume the $4$-connectivity of the solution.

\begin{figure}[!b]
\vspace*{-1mm}
  \begin{center}
		\includegraphics[width=0.85\textwidth]{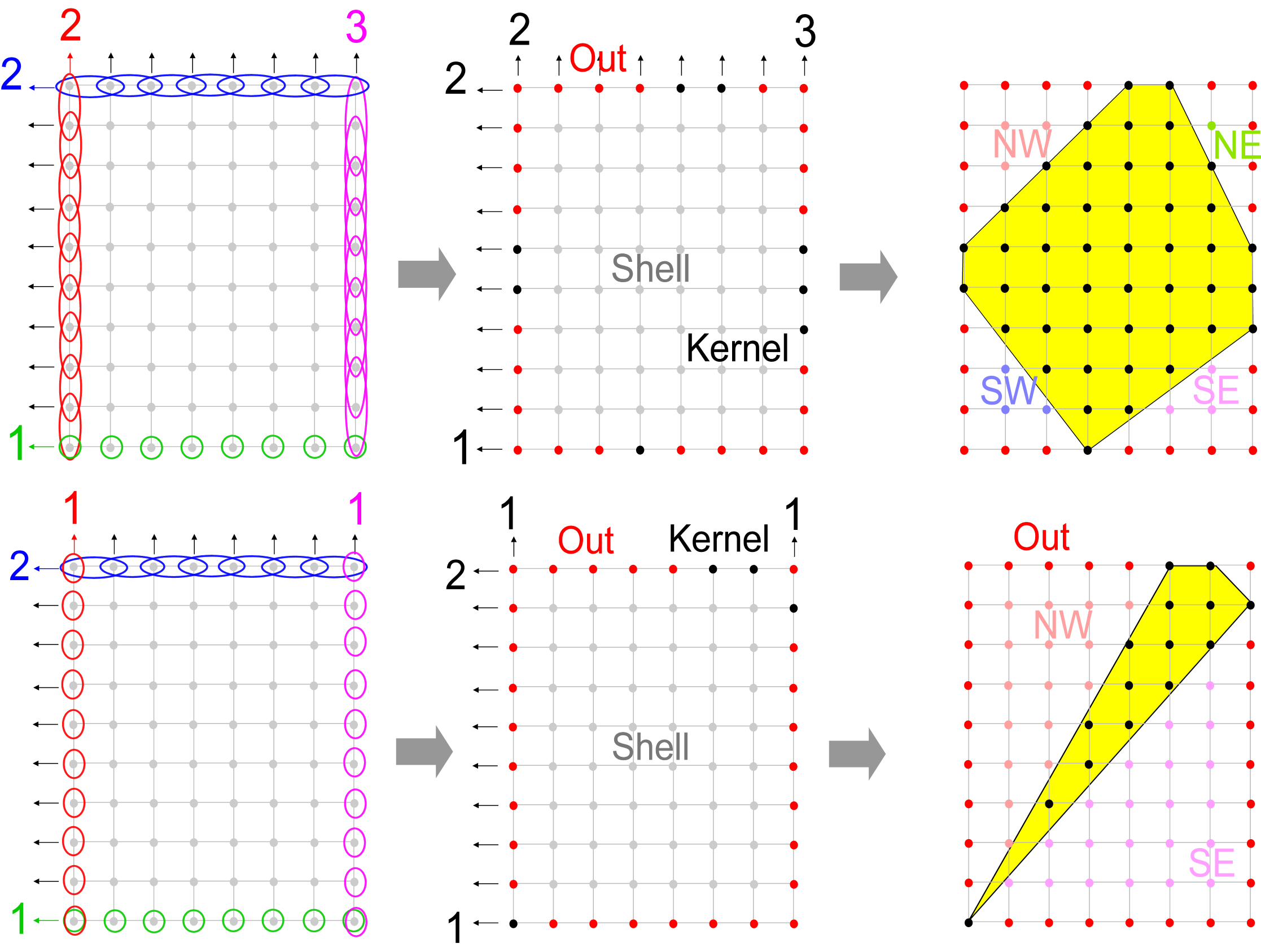}
	\end{center}\vspace*{-4mm}
	\caption{\label{fig-3} \textbf{Initialization of $\In$,  $\Out$, $\NW$, $\NE$, $\SE$ and $\SW$ for two different instances}. On the left, the different possible positions of the feet. A position of the feet being chosen (in black), the other points with $x=0$ or $x=m-1$ or $y=0$ or $y=n-1$ are added to $\Out$ (in red). The undetermined points are colored in grey. On the right, we proceed to the filling operation $\In \leftarrow \conv (\In)$. It provides directly a partition of the undetermined points  in four subsets $\NW$, $\NE$, $\SE$ and $\SW$.
}
\end{figure}

\paragraph{New Convex Filling Operations.}

As we search for a digital convex solution, the first filling operation is
$\In \leftarrow \conv (\In)$ where $\conv$ denotes the operator providing the set of the integer points in the continuous convex hull of the set. By definition, a set  $S$ is digital convex if and only if we have $\conv(S)=S$.
As the convex hull of the kernel contains the feet, it provides directly a partition of the undetermined points in four borders
$\NW$, $\NE$, $\SE$ and $\SW$ (see Fig.~\ref{fig-3}). While this partition requires to use the connectivity and complex filling operations for reconstructing HV-convex polyominoes, it is direct for digital convex sets.
The initial computation of the four sets $\NW$, $\NE$, $\SE$ and $\SW$ takes $O(mn)$ operations.

Each time that a new point is added to the kernel, we update the kernel with \\ 
 $\In \leftarrow \conv (\In)$ to maintain a convex solution. The number of nodes is bounded by $2(m+n)$. Computing the new convex hull takes at most $O(\log (m+n))$ (it is an insertion in the list of the nodes) and adding the lattice points in $\conv (\In)$ takes again $O(\log (m+n))$ per point namely a total time of $O(mn \log (m+n))$ (more efficient data structures can be found in  \cite{Brodal}). By repeating this operation at most $mn$ time, the total time taken by this filling operation is bounded by $O(m^2n^2 \log (m+n))$.

\medskip
Each time that a new point $x$ is added to the set $\Out$ of the excluded points, we test whether there exist undetermined points $y$ such that  $x$ is in the convex hull of the union $\{y\} \cap \Out$. If $x$ is in $\{y\} \cap \Out$, then $y$ is excluded and added to $\Out$.
For each point $x$ and each suspected point $y\in \Undetermined$, the computation of the convex hull of $\{y\} \cap \Out$ takes at most $O(\log (m+n))$ time. Testing whether $x$ is in this convex hull takes a constant time by just testing whether $x$ is in the triangle of vertices $y$ and the two adjacent vertices in the convex hull. Repeating this process for at most $mn$ points $x$ takes  $O(m^2n^2 \log (m+n))$ time.

\paragraph{Regular Filling Operations.}

We describe the classical filling operations on a row but similar filling operations are executed on the columns. The row $y=j$ is decomposed into five horizontal segments with their $x$-coordinates in the intervals
$[0 \cdots a]$, $]a \cdots b[$, $[b \cdots c]$, $]c \cdots d[$ and $[d \cdots m-1]$ (some of them can be empty). The two segments with points of $x$-coordinates in $[0 \cdots a]$ and $[d \cdots m-1]$ contain the left and right excluded points. The central segment $[b \cdots c]$ contains the kernel points while  the two intermediate intervals $]a \cdots b[$ and $]c \cdots d[$ are the $x$ coordinates of the undetermined points.

Then the regular filling operations are done by updating the values of $a$, $b$, $c$ and $d$ in the following way :
\begin{itemize}
\itemsep=0.96pt
    \item if $a<c-h_j-1$ then $a\leftarrow c-h_j-1$
    \item if $b>d-h_j$ then $b\leftarrow d-h_j$,
    \item if $c<a+h_j$ then $c\leftarrow a+h_j$,
    \item if $d>b+h_j$ then $d\leftarrow b+h_j$
\end{itemize}

These four filling operations reduce the set of the undetermined points.

\paragraph{End of the Filling Step.}

The filling step ends in three cases:
\begin{enumerate}
\itemsep=0.98pt
    \item A contradiction is found by adding an excluded point to the kernel or conversely by adding a kernel point to the excluded points. Then the reconstruction procedure fails. It shows that the considered position of the feet does not provide any solution.
    \item The set of the undetermined points becomes empty. Then for the considered position of the feet, a unique solution has been found and the reconstruction procedure ends with a unique solution.
    \item No new undetermined point can be determined. It is the unique case where the reconstruction procedure requires to proceed to the aggregation step.
\end{enumerate}

\subsection{Aggregation step}\label{gla}

We proceed to the aggregation step in the case where it remains undetermined points after the filling step. The aggregation step for the HV-convex polyominoes uses the switching components. The aggregation step of \texttt{DAGTomo2} does not use it. It uses the same strategy as the aggregation step of \texttt{DAGTomo1} by reducing the computation of a solution to the research of path in Directed Acyclic Graph.\vspace*{-1mm}

\paragraph{The Four Undetermined Borders.}
The goal of the convex aggregation is to finalize the computation of a  solution of $DT_{\C \cap \F} (h,v)$. This solution can be described by the convex polygonal lines joining the $\West$ and $\East$ feet to the $\South$ and $\North$ feet. The vertices of the polygonal line going from the $\West$ to the $\North$ feet might either be some vertices of the convex hull of the kernel or some undetermined points in the $\NW$ border. This set of potential vertices is denoted $\NW'$ with $\NW '= \NW \cap \NWV$ where $\NWV$ is the set of the north west vertices of the convex hull of the kernel i.e the vertices between the North and the West feet. The set $\NWV$ and thus also $\NW'$ contain the upper point of the $\West$ foot and the leftmost point of the $\North$ foot.

\begin{figure}[!h]
\vspace*{2mm}
  \begin{center}
		\includegraphics[width=0.45\textwidth]{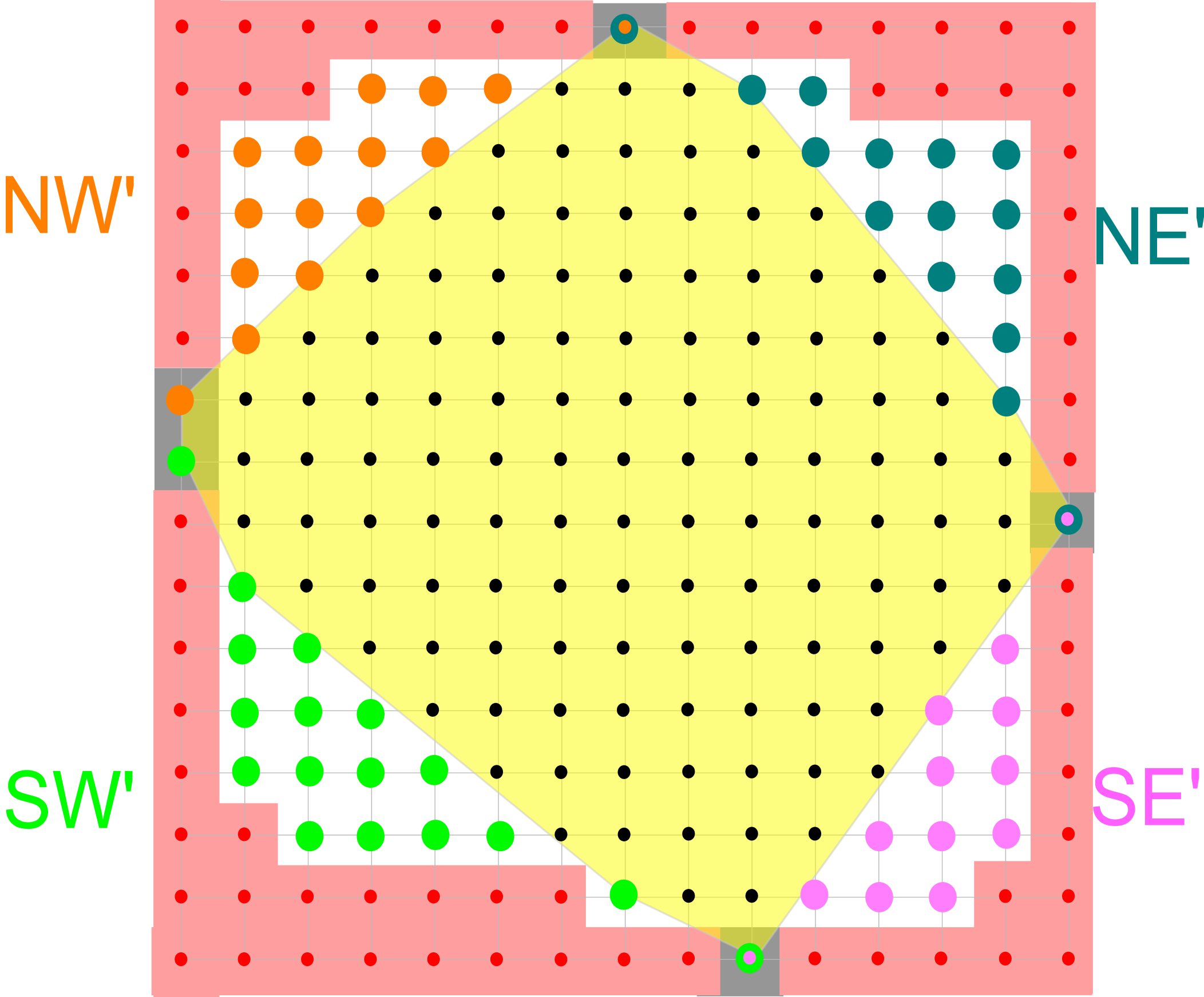}
	\end{center}\vspace*{-2mm}
	\caption{\label{4borders} \textbf{The four borders}. After the filling step, the convex hull of the kernel is colored in yellow. The set of the potential vertices $\NE '$, $\NW '$, $\SE '$ and $\SW '$  contain the undetermined points of each border and the vertices of the convex hull of the kernel (the example shown above is not realistic since the filling operations would complete the kernel but this non-realistic example provides a better understanding of these sets).
	The points of these four sets are represented by  blue, orange, pink and green dots.
	Notice that two sets of potential vertices share a point when a feet contains a unique point.
}
\end{figure}

\medskip
We define in the same manner the set of the potential vertices for the North East, South West and South East borders as $\NE '= \NE \cap \NEV$, $\SW '= \SW \cap \SWV$, and $\SE '= \SE \cap \SEV$ where $\NEV$, $\SWV$ and $\SEV$ are the vertices of the convex hull of the kernel on each border. Each set of potential vertices contains the extreme points of the feet (Fig.~\ref{4borders}).

\paragraph{The Central Fully Determined Strips.}

We consider here the reconstruction of fat digital convex sets. It requires that the feet are in relative positions leading to a fat solution. In this case, the four undetermined borders are in a special configuration. According to the analysis of the configurations of the switching components presented in figure \ref{abcdef}\cite{RSC}, the borders containing the undetermined points are restricted to only four regions corresponding vertically and horizontally as drawn in Fig.~\ref{critics}.
We denote $\Hstrip$ the horizontal central strip containing the $\West$ and $\East$ feet and whose points are fully determined. The points in this strip belong either to the kernel or have been excluded but no one is undetermined.
In the same way, we denote $\Vstrip$ the central vertical strip containing the $\South$ and $\North$ feet and whose points are fully determined.
These two strips define a cross and the four undetermined borders are around it. 

\begin{figure}[ht]
\vspace*{2mm}
  \begin{center}
		\includegraphics[width=0.75\textwidth]{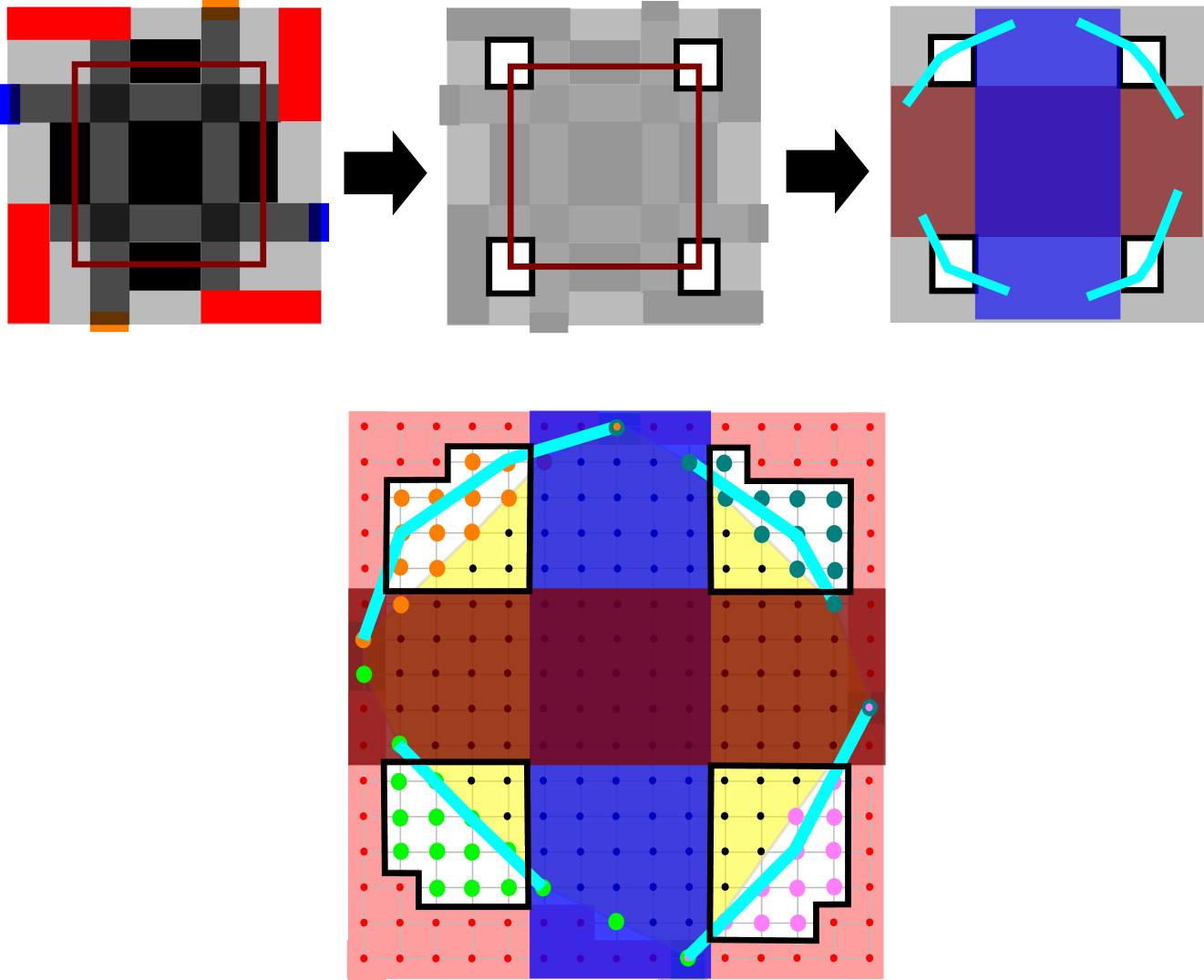}
	\end{center}\vspace*{-3mm}
	\caption{\label{critics} \textbf{The region to fix}. In the case of feet in a fat position, the four undetermined borders belong to four regions on different sides of the horizontal strip $\Hstrip$ represented in brown and the vertical strip $\Vstrip$ represented in blue. The problem of reconstruction can be reduced to the computation of the polygonal lines cutting the four undetermined regions represented in white.
}
\end{figure}

  \begin{figure}[!b]
  \begin{center}
		\includegraphics[width=0.84\textwidth]{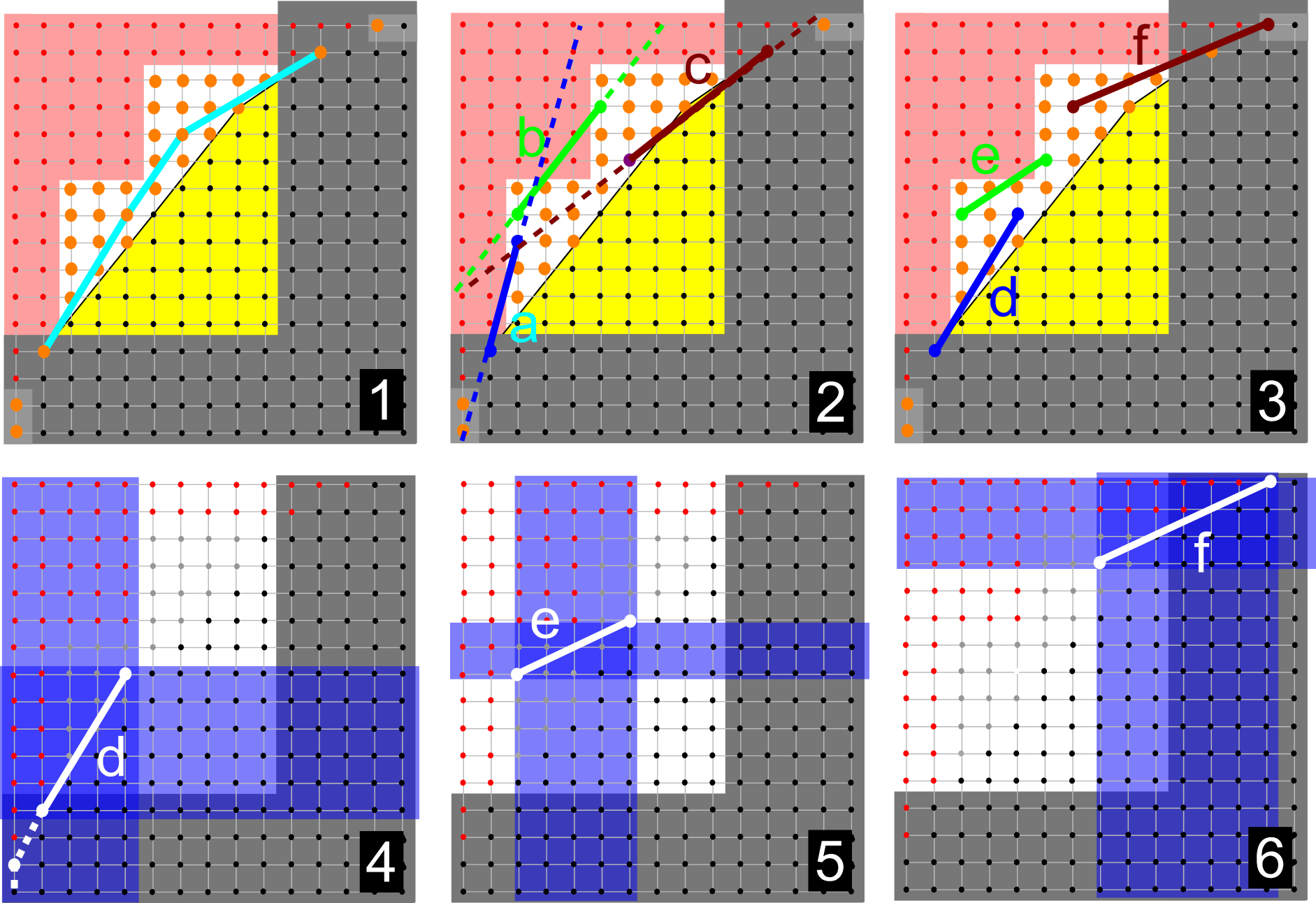}
	\end{center}\vspace*{-2mm}
	\caption{\label{validsegments} \textbf{Some segments and their associated rows and columns}.  In image 1, we see the $\NW$ border with potential vertices in orange and the horizontal and vertical  fully determined strips $\Hstrip$ and $\Vstrip$ colored in grey. In image 2, the segments a, b, and c are not valid because the support lines of a and c separate the kernel while b includes an excluded point. In image 3, we have three valid segments $s_ {\NW}$. In the images 4, 5, 6 the rows of $rows(d)$, $rows(e)$ and $rows(f)$ and the columns of $columns(d)$, $columns(e)$ and $columns(f)$ are colored in blue. Notice that since $d$ has a vertex in the horizontal strip $\Hstrip$, its set of columns is extended until the $\West$ foot.
}
\end{figure}

\paragraph{The DAG.} We introduce the Directed Acyclic Graph $\Gamma(h,v)$ built from the $h$ and $v$ X-rays and from the four sets of potential vertices $\NW '$, $\NE '$, $\SW '$, and $\SE '$ as follows.
\begin{itemize}
    \item The nodes of $\Gamma(h,v)$ are octagons  defined with four (non degenerated) segments $s_{\NW}$, $s_{\NE}$, $s_{\SW}$, $s_{\SE}$ satisfying some constraints  (Fig.~\ref{octogon}). The role of each segment $s_{\NW}, s_{\NE}, s_{\SW}, s_{\SE}$ is to become an edge of the final solution.
    The octagons $q=s_{\NW}, s_{\NE}, s_{\SW}, s_{\SE}$ that we describe here with  their segments and constraints are the central piece of the algorithm \texttt{DAGTomo2}.
    \begin{enumerate}
        \item The vertices of the segments $s_{\NW}, s_{\NE}, s_{\SW}, s_{\SE}$ belong respectively to the sets of potential vertices $\NW '$, $\NE '$, $\SW '$, $\SW '$ that we introduced in the previous paragraphs.
        \item The support line of each segment $s$ must be on one side of the kernel and not separate it into two parts. Otherwise, the segment $s$ can not be an edge of a convex set containing the kernel.
        \item The octagon (here the convex hull of the four segments) does not contain any excluded point.
  \end{enumerate}

   \begin{figure}[!ht]
    \begin{center}
		\includegraphics[width=0.7\textwidth]{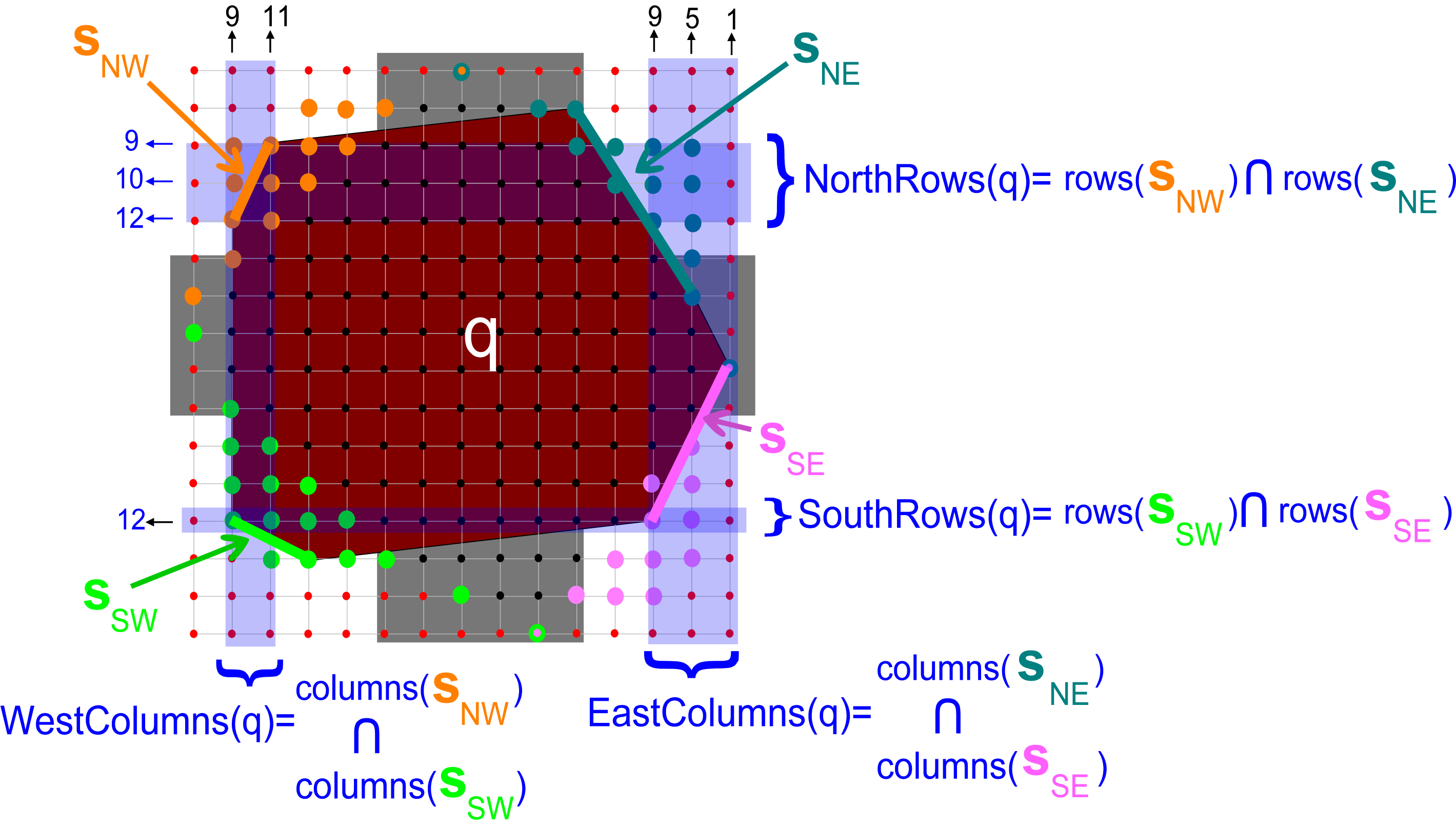}
	\end{center}\vspace*{-5mm}
	\caption{\label{octogon} \textbf{An octagon/node of $\Gamma(h,v)$}. The octagon $q$ is defined by its four non degenerated segments $s_{\NW}, s_{\NE}, s_{\SW}, s_{\SE}$  whose vertices are chosen respectively in $\NW '$, $\NE '$, $\SW '$ and $\SE '$ (the orange/blue/green/pink dots).
	The main conditions for the octagon $q$ to be a node are that the rows
	$NorthRows(q)$,  $SouthRows(q)$ and the columns $WestColumn(q)$, $RightColumn(q)$ are non empty and that for each one of the rows and columns in these sets, the number of integer points between the segments fits with the prescribed coordinates of the X-rays.
	On this example, it can be noticed that columns $columns(s_ {\NE})$ of the $\NE$ edge is not restricted to the columns crossed by the segment. They go until the $\West$ foot because both right vertices of $s_ {\NE}$ and $s_{\SE}$ are in the horizontal strip $\Hstrip$ (colored in grey).
}
\end{figure}

 \begin{itemize}
        \item[{4.}]  The next condition guarantees that the X-ray of the solution fits with the prescribed X-rays $h$ and $v$.
        We first associate to each segment $s_{\NW}, s_{\NE}, s_{\SW}, s_{\SE}$ the set of consecutive rows and columns that it involves. For a segment $s$, these two sets are denoted $rows(s)$ and $columns(s)$.
        The sets $rows(s)$ and $columns(s)$ first contain the consecutive integer rows and columns crossed by $s$ (Fig.~\ref{validsegments}). If the segment $s_{\NW}=xy$ has one vertex $x$ in the horizontal strip $\Hstrip$, then we can extend the known polygonal line until the $\West$ foot since this left upper hull of the solution is already known until the $\West$ foot. It leads us to complete $columns(s_{\NW})$ with its left columns. We proceed in the same way with the vertical strip. If the segment $s_{\NW}=xy$ has one vertex $y$ in the vertical strip $\Vstrip$, then we complete $rows(s_{\NW})$ with the rows above it. By symmetry and rotations, we extend in the same manner the six other sets $rows(s_{\NE})$, $columns(s_{\NE})$, $rows(s_{\SE})$, $columns(s_{\SE})$, $rows(s_{\SW})$, $columns(s_{\SW})$.

         Once the rows and columns associated to each segment are defined, we can go to the octagon. Each octagon is determined by its four segments $s_{\NW}, s_{\NE}, s_{\SW}, s_{\SE}$ whose role is to become four edges of the convex hull of a solution. Then the octagon $q$ fully characterizes the points of the solution in the rows in
        $NorthRows(q)=rows(s_{\NW}) \cap rows(s_{\NE})$ and $SouthRows(q)=rows(s_{\SW}) \cap rows(s_{\SE})$ and in the columns in $WestColumn(q)=columns(s_{\NW}) \cap columns(s_{\SW})$, and $RightColumn(q)=columns(s_{\NE}) \cap$\linebreak  $columns(s_{\SE})$ (these rows and columns are colored in blue in Fig.~\ref{octogon}).
        Now we give a condition on the octagon $q$ to be a node: the four sets of rows
        $NorthRows(q)$, \linebreak $SouthRows(q)$ and columns $WestColumn(q)$, $RightColumn(q)$ must be non empty and for each of these rows and columns, the number of integer points between each pair of segments must be equal to the prescribed coordinate of the input X-ray.
    \end{itemize}

    \item We consider two nodes. There is an edge from the node $q$ to the node $q'$ if they have three common segments and the two segments which differ define a convex turn. To be more precise, they must satisfy the following  conditions:
    \begin{itemize}
        \item The two nodes $q$ and $q'$ have three segments in common. Let us assume for instance that we have the two nodes $q=s_{\NW}, s_{\NE}, s_{\SW}, s_{\SE}$ and $q'=s'_{\NW}, s_{\NE}, s_{\SW}, s_{\SE}$. The three segments $ s_{\NE}, s_{\SW}, s_{\SE}$ are in common and the two segments $s_{\NW}$ and $s_{\NW}'$ are different.
        \item The two different segments $s_{\NW}$ and $s_{\NW}'$ are consecutive (the right vertex of $s_{\NW}$ is the left vertex of $s_{\NW}'$).
        \item The turn at the common vertex is convex (in the sense that the union of the convex hulls of the two octagons $q= s_{\NW}, s_{\NE}, s_{\SW}, s_{\SE}$ and $q'=s_{\NW}, s_{\NE}, s_{\SW}, s'_{\SE}$ is convex).
    \end{itemize}
\end{itemize}

\paragraph{Start and End nodes.}

We introduce now two special sets of octagons $\Start$ and $\End$ (Fig.~\ref{startend}).

\begin{figure}[!h]
\vspace*{4mm}
  \begin{center}
		\includegraphics[width=0.75\textwidth]{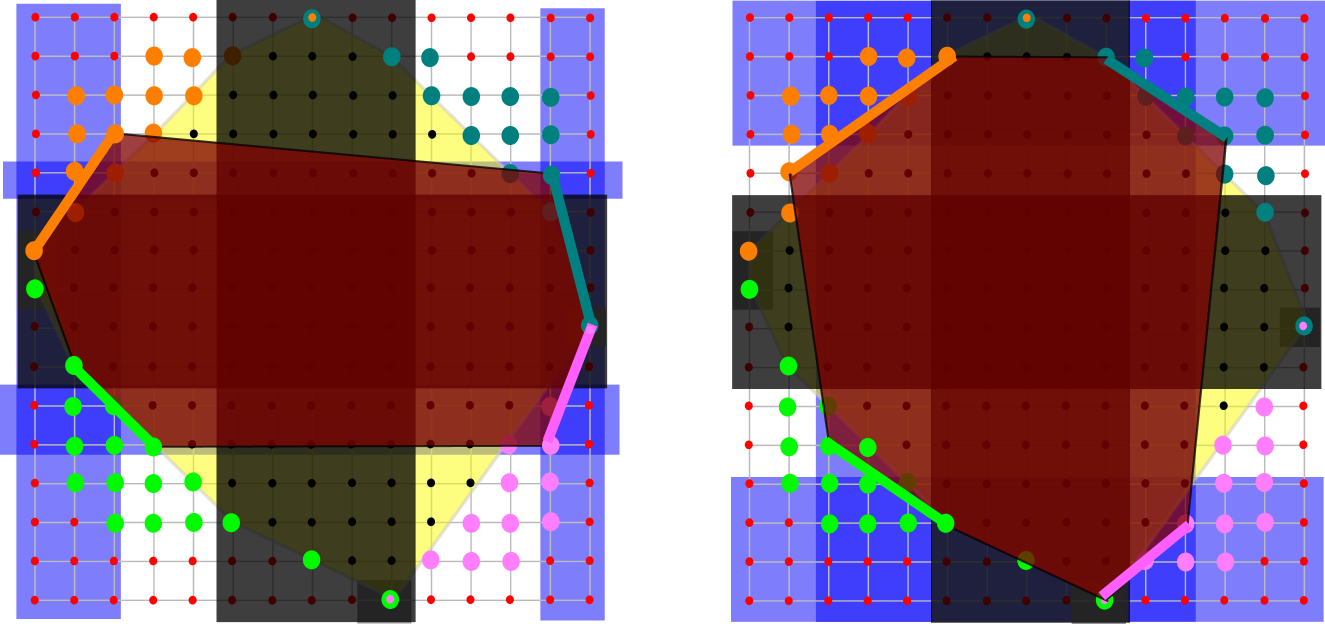}
	\end{center}\vspace*{-3mm}
	\caption{\label{startend} \textbf{Start and End octagons}. On the left, an octagon is in $\Start$ if the four segments $s_{\NW}$, $s_{\NE}$, $s_{\SW}$ and $s_{\SE}$ have a vertex in the horizontal strip (colored in dark grey) and another vertex outside. On the right, an octagon is in $\End$  if the four segments $s_{\NW}$, $s_{\NE}$, $s_{\SW}$ and $s_{\SE}$ have a vertex in the vertical strip (colored in dark grey) and another vertex outside. Notice that the octagon in $\Start$ is degenerated because the south segments have a common vertex. This degenerated case might arise when a foot contains a unique point.
}\vspace*{-3mm}
\end{figure}

\begin{itemize}
    \item $\Start$ is the set of nodes/octagons with the property that the four segments $s_{\NW}$, $s_{\NE}$, $s_{\SW}$ and $s_{\SE}$ have one vertex in the horizontal strip $\Hstrip$ and the other vertex outside.
\item $\End$ is the set of nodes/octagons such  that the four segments $s_{\NW}$, $s_{\NE}$, $s_{\SW}$ and $s_{\SE}$ have one vertex in the vertical strip $\Vstrip$ and the other vertex outside.
\end{itemize}

Notice the special configurations of the rows $NorthRows(q)$,  $SouthRows(q)$ and columns \\ $WestColumn(q)$, $RightColumn(q)$  for octagons in $\Start$ and $\End$. For an octagon in $\Start$, the determined rows and columns do a H while for the an octagon, they do an I (with large bars below and above the vertical segment).
If we now consider a continuous transformation of an octagon $q$ going from $\Start$ to $\End$, its  rows $NorthRows(q)$,  $SouthRows(q)$ and columns $WestColumn(q)$, $RightColumn(q)$ sweep the whole intermediate rows and columns. It is the strategy followed by the algorithm \texttt{DAGTomo2} in order to guarantee that the X-rays of a solution correspond to the prescribed X-rays.

\paragraph{Reduction.}

We reduce the problem of convex aggregation for solving $DT_{\C \cap \F} (h,v)$ to the research of a path going from $\Start$ to $\End$ in the graph
$\Gamma(h,v)$.

\begin{proposition}\label{redu}
The problem $DT_{\C \cap \F} (h,v)$ with some considered positions of the feet $\South$, $\West$, $\North$ and $\East$ admits a solution if and only if after the filling step, there exists a path from $\Start$ to $\End$ in the graph $\Gamma(h,v)$.
\end{proposition}

\begin{proof}
We prove Proposition \ref{redu} by showing first that if $DT_{\C \cap \F} (h,v)$ admits a solution, then there is a path  from $\Start$ to $\End$ in the graph $\Gamma(h,v)$ and secondly  by stating that the existence of a path from $\Start$ to $\End$ provides a solution of $DT_{\C \cap \F} (h,v)$.

We first assume first that $DT_{\C \cap \F} (h,v)$ admits a solution $S$ with the considered position of the feet. The filling step provides the set of the undetermined points. According to \cite{RSC}, the undetermined points are organized in four borders $\NW$, $\NE$, $\SE$ and $\SW$ separated by a vertical strip $\Vstrip$ and a horizontal strip $\Hstrip$ of fully determined rows and columns.
Then the undetermined part of the solution $S$ can be decomposed in a sequence of octagons $q_i$ for $0\leq i \leq k$ as drawn in Fig.~\ref{soltopath}. Different sequences of consecutive octagons can be obtained.  It depends on the order chosen to move the segments. Each one provides a path from $\Start$ to $\End$.

\begin{figure}[!h]
\vspace*{2mm}
  \begin{center}
		\includegraphics[width=0.71\textwidth]{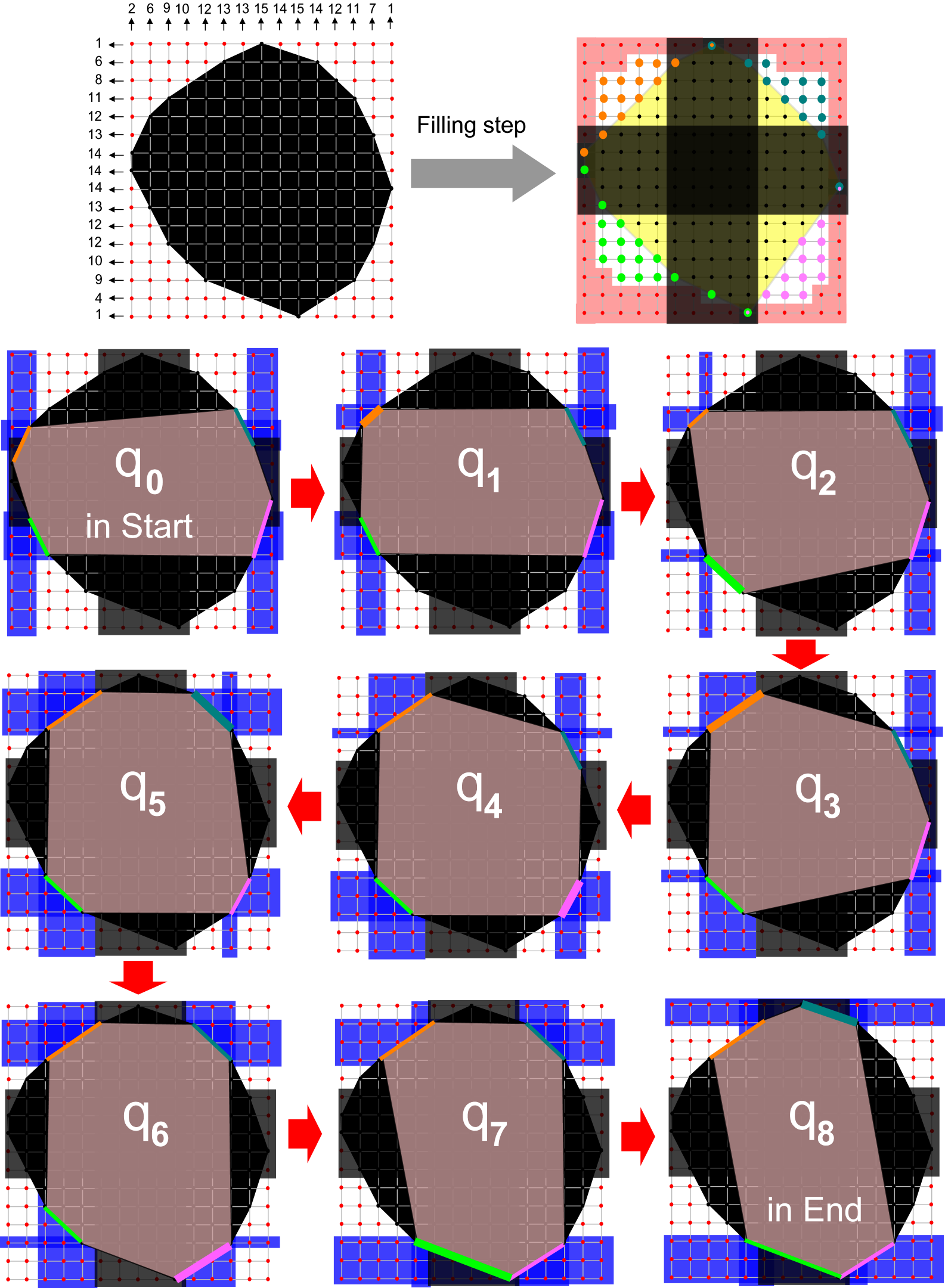}
	\end{center}\vspace*{-2mm}
	\caption{\label{soltopath} \textbf{From a solution to a path in $\Gamma (h,v)$.}}\vspace*{-5mm}
\end{figure}

\medskip
Conversely, we assume that we have a path of consecutive octagons $(q_i)_{1\leq i \leq k}$ in $\Gamma (h,v)$ with its extreme octagons $q_0$ in $\Start$ and $q_k$ in $\End$. We define the set $O$ as the union of two sets. The first one contains the integer points in the octagons of the path. The second one is the kernel. Then we have  $O= \bigcup _{0\leq i \leq k}q_i \cup \In$. We now show that $O$ is digital convex and has the prescribed X-rays.

The kernel is digital convex since one of the filling operations is $\In \leftarrow \conv (\In)$.
The set $\bigcup _{0\leq i \leq k}q_i$ is also digital convex. It follows from the condition of having a convex turn between consecutive octagons in the graph $\Gamma (h,v)$.
 We notice also that except in the horizontal and vertical strips $\Hstrip$ and $\Vstrip$, the convex hull of the kernel is included in the union of the octagons $\bigcup _{0\leq i \leq k}q_i$. In other words, the union of the octagons is a convex layer around the  kernel outside from $\Hstrip$ and $\Vstrip$.

 The unique points where we could have a non-convexity are the vertices of the convex hull of the kernel where the octagons $q_0$ and $q_k$ are branched. The condition that the support lines of the segments of $q_0$ and $q_k$ do not separate the kernel guarantees that the junction is convex. Then the set $O$ is digital convex.

We prove now that the horizontal and vertical X-rays of $O$ are equal to the prescribed X-rays $h$ and $v$. It is straightforward for the rows and columns of the two central strips $\Hstrip$ and $\Vstrip$ since there is no undetermined point in this region. For the remaining rows and columns, each octagon $q_i$ guarantees the X-rays for the rows  $NorthRows(q_i)$,  $SouthRows(q_i)$ and columns $WestColumn(q_i)$, $RightColumn(q_i)$. It remains to notice that by going from $q_0$ to $q_k$, these guaranteed rows and columns cover all the raws and columns which are not in the already determined strips $\Hstrip$ and $\Vstrip$.
\end{proof}

\paragraph{\texttt{DAGTomo2} Aggregation Algorithm.}

The \texttt{DAGTomo2} aggregation algorithm is based on Proposition \ref{redu} reducing the aggregation problem to the research of a path from $\Start$ to $\End$ in the graph $\Gamma(h,v)$. The strategy is simple:
\begin{enumerate}
    \item Compute the octagons in $\Start$.
    \item Compute the octagons which can be reached from $\Start$  with a Depth first search (each time a new octagon is reached, the algorithm records its previous neighbor so that at the end we can backtrack the path).
\end{enumerate}

The algorithm admits two variants which depend on which part of the graph $\Gamma(h,v)$ is precomputed and stored. In any case, the algorithm requires to store a set of octagons reached from $\Start$. The size of this set is $O(m^4n^4)$.
In the variant that we consider, the vertices of the graph $\Gamma(h,v)$ are precomputed but not the edges.
The relations between the vertices are computed on the fly.
The precomputation of the vertices of $\Gamma(h,v)$ is done by considering all the potential octagons. We have $O(m^4n^4)$  quadruplets of segments $s_\NW$, $s_\NE$, $s_ \SE$, $s_\SW$. Testing whether a quadruplet of segments defines a valid octagon requires to test two criterion:
\begin{enumerate}
\item Test whether the segments do not include an excluded point. It takes $O(m+n)$ time by considering only the extreme excluded points.
\item Compare the X-rays of the octagon for the rows and columns that it fully determines with the input X-rays $h$ and $v$. This comparison takes $O(m+n)$.
\end{enumerate}
Then the construction of the set of vertices of $\Gamma(h,v)$ takes $O((m+n)m^4n^4)$.

Given an octagon $q$, the computation of its neighbors $q'$ is done on the fly in the following way:
The neighbors of an octagon are obtained by moving only one of its vertices. The number of new vertices is $O(mn)$. It makes $O(mn)$ potential neighbors $q'$ for which we test the convexity of the union $q\cup q'$ and whether $q'$ is a vertex of $\Gamma(h,v)$. It takes a constant time.

\subsection{Analysis}

We now analyze the worst-case time complexity of the algorithm \texttt{DAGTomo2}.
The number of feet positions being bounded by $O(m^2n^2)$, the algorithm executes at most $O(m^2n^2)$ reconstruction procedures.

The reconstruction procedure starts with the filling step. The total time of the new filling operations related to the digital convexity is $O(m^2n^2 \log (m+n))$. The total time of the other filling operations is $O(mn)$. It provides a bound $O(m^2n^2 \log (m+n))$ for the filling step.

The reconstruction procedure ends with the aggregation step. The time complexity of the Depth first search for a given graph is $O(|V|+|E|)$ where $|V|$ is the total number of vertices  and $|E|$ the total number of edges of the graph $\Gamma(h,v)$. The complexity of the algorithm is not affected by the computation on the fly of the edges since it is done in constant time. As we have $|V|=O(m^4n^4)$ and at most $O(mn)$ neighbors per vertex, we have $|V|=O(m^5n^5)$.
By repeating the reconstruction procedure at most $m^2n^2$ times, the complexity of \texttt{DAGTomo2} is
$O(m^7n^7)$. It proves Theorem \ref{t1}.

\section*{Perspectives}

We have shown that the DAG strategy used in the algorithms \texttt{DAGTomo2} and \texttt{DAGTomo2} allows us to  solve $DT_{\C } (v)$ and $DT_{\C \cap \F} (h,v)$ in polynomial time. It is a step beyond the previously known results.
Can it be used to fully solve the challenging problem $DT_{\C} (h,v)$ or at least some other configurations of the feet than the fat configuration f) of Fig.~\ref{cycles}? The complexity of the reconstruction of digital convex lattice sets from a pair of X-rays remains an open question but there is still hope that new results can be derived from the ideas presented in this paper or from other breakthroughs.

\end{document}